\documentclass{article}

     \usepackage[final, nonatbib]{neurips_2020}

\usepackage[numbers]{natbib}
\usepackage[utf8]{inputenc} 
\usepackage[T1]{fontenc}    
\usepackage[colorlinks]{hyperref}       
\usepackage{url}            
\usepackage{booktabs}       
\usepackage{amssymb,amsmath,amsfonts,amsthm}       
\usepackage{nicefrac}       
\usepackage{microtype}      
\usepackage[colorinlistoftodos,prependcaption,textsize=tiny]{todonotes}

\usepackage{algorithm}
\usepackage[noend]{algpseudocode}
\usepackage{thmtools}
\usepackage{thm-restate}
\usepackage{pgfplots}
\usepackage{xspace}
\usepackage{bbm}
\usepackage{booktabs}
\usepackage{longtable}
\usepackage{array}
\usepackage{caption}
\usepackage{graphicx}

\newtheorem{theorem}{Theorem}
\newtheorem{lemma}{Lemma}
\newtheorem{definition}{Definition}
\theoremstyle{remark} 
\newtheorem{remark}{Remark}

\DeclareMathOperator*{\E}{\mathrm{E}}

\newcommand{\ceil}[1]{\left\lceil #1\right\rceil}
\newcommand{\accept}{$\mathtt{ACCEPT}$}
\newcommand{\reject}{$\mathtt{REJECT}$}

\newcommand{\wt}[1]{\ensuremath{\mathtt{wt} \! \left( #1 \right)}}

\newcommand{\bias}{\ensuremath{\mathsf{Bias}}}
\newcommand{\kernel}{\ensuremath{\mathsf{Kernel}}}

\newcommand{\SearchTreeSampler}{\ensuremath{\mathsf{STS}}}
\newcommand{\WeightGen}{\ensuremath{\mathsf{UniGen}}}
\newcommand{\oldbarbarik}{\ensuremath{\mathsf{Barbarik}}}
\newcommand{\fulcrumkernel}{\ensuremath{\mathsf{Barbarik2Kernel}}}
\newcommand{\barbarik}{\ensuremath{\mathsf{Barbarik2}}}
\newcommand{\QuickSampler}{\ensuremath{\mathsf{Quicksampler}}}
\newcommand{\wQuickSampler}{\ensuremath{\mathsf{wQuicksampler}}}
\newcommand{\wSearchTreeSampler}{\ensuremath{\mathsf{wSTS}}}
\newcommand{\wWeightGen}{\ensuremath{\mathsf{wUniGen}}}
\newcommand{\kernelfamily}{\ensuremath{\mathsf{KernelFamily}}}

\title{On Testing of Samplers \thanks{The accompanying tool, available open source, can be found at \href{https://github.com/meelgroup/barbarik}{https://github.com/meelgroup/barbarik}. The Appendix is available in the accompanying supplementary material.} \thanks{The authors decided to forgo the old convention of alphabetical ordering of authors in favor of a randomized ordering, denoted by \textcircled{r}. The publicly verifiable record of the randomization is available at \protect\url{https://www.aeaweb.org/journals/policies/random-author-order/search} with confirmation code: GH8VZdz4mQIh. For citation of the work, authors request that the citation guidelines by AEA for random author ordering be followed. }}

\author{%
  Kuldeep S. Meel$^1$ \quad \textcircled{r} \quad   Yash Pote $^1$ \quad \textcircled{r}  \quad
  Sourav Chakraborty$^2$ \\
  $^1$School of Computing, National University of Singapore\\
  $^2$Indian Statistical Institute, Kolkata
}

\begin{document}

\maketitle
\begin{abstract}
	
	Given a set of items $\mathcal{F}$ and a weight function $\mathtt{wt}: \mathcal{F} \mapsto (0,1)$, the problem of sampling seeks to sample an item proportional to its weight. Sampling is a fundamental problem in machine learning. The daunting computational complexity of sampling with formal guarantees leads designers to propose heuristics-based techniques for which no rigorous theoretical analysis exists to quantify the quality of generated distributions. 
	This poses a challenge in designing a testing methodology to test whether a sampler under test generates samples according to a given distribution.     Only recently, Chakraborty and Meel (2019) designed the first scalable verifier, called {\oldbarbarik}, for samplers in the special case when the weight function $\mathtt{wt}$ is constant, that is, when the sampler is supposed to sample uniformly from $\mathcal{F}$ . The techniques in {\oldbarbarik}, however,  fail to handle general weight functions. 
	
	The primary contribution of this paper is an affirmative answer to the above challenge: motivated by {\oldbarbarik}, but using different techniques and analysis, we design {\barbarik}, an algorithm to test whether the distribution generated by a sampler is $\varepsilon$-close or $\eta$-far from any target distribution. In contrast to black-box sampling techniques that require a number of samples proportional to $|\mathcal{F}|$ , {\barbarik} requires only $\tilde{O}(tilt(\mathtt{wt},\varphi)^2/\eta(\eta - 6\varepsilon)^3)$ samples, where the $tilt$ is the maximum ratio of weights of two satisfying assignments. {\barbarik} can handle any arbitrary weight function. We present a prototype implementation of {\barbarik} and use it to test three state-of-the-art samplers.
\end{abstract} 
\section{Introduction}
Motivated by the success of statistical techniques, automated decision-making systems are increasingly employed in critical domains such as medical~\cite{CK19}, aeronautics~\cite{MS18}, criminal sentencing~\cite{D19}, and military~\cite{AS17}. The potential long-term impact of the ensuing decisions has led to research in the correct-by-construction design of AI-based decision systems. There has been a call for the design of randomized and quantitative formal methods~\cite{SSS16} to verify the basic building blocks of the modern AI systems. In this work, we focus on one such core building block: {\em constrained sampling}.

Given a set of constraints $\varphi$ over a set of variables $X$ and a weight function $\mathtt{wt}$ over assignments to $X$, the problem of constrained sampling is to sample a satisfying assignment $\sigma$ of $\varphi$ with probability proportional to $\mathtt{wt}(\sigma)$.  Constrained sampling is a fundamental problem that encapsulates a wide range of sampling formulations~\cite{GSS07,EGSS13a,CMV13a,MTM14,CFMSV15}. For example, $\mathtt{wt}$ can be used to capture a given prior distribution often represented implicitly through probabilistic models, and $\varphi$ can be used to capture the evidence arising from the observed data, then the problem of constrained sampling models the problem of sampling from the resulting posterior distribution.

The problem of constrained sampling is computationally hard and has witnessed a sustained interest from theoreticians and practitioners, resulting in the proposal of several approximation techniques. Of these, Monte Carlo Markov Chain(MCMC)-based methods form the backbone of modern sampling techniques~\cite{ADDJ03,BGJM11}. The runtime of these techniques depends on the length of the random walk, and the Markov chains that require polynomial walks are called rapidly mixing Markov chains. Unfortunately, for most distributions of practical interest, it is infeasible to design rapidly mixing Markov chains~\cite{J98}, and the practical implementations of such techniques have to resort to the usage of heuristics that violate theoretical guarantees. The developers of such techniques, often and rightly so, strive to demonstrate their effectiveness via empirical behavior in practice~\cite{BG98}. 

The need for the usage of heuristics to achieve scalability is not restricted to just MCMC methods but is widely observed for other methods such as simulated annealing~\cite{KGV83}, variational methods~\cite{CDAD04}, and hashing-based techniques~\cite{CMV13a,EGSS13a,CFMSV14,MVC+16}. Consequently, a fundamental problem for the designers of sampling techniques is:
{\em how can one efficiently test whether a given technique samples from the desired distribution?}
Most of the existing approaches rely on the computations of statistical metrics such as variation distance and KL-divergence by drawing samples and perform hypothesis testing with a preset p-value. Sound computations of statistical metrics require a large number of samples that is proportional to the support of the posterior distribution \cite{batu, ValiantV11}, which is prohibitively large; it is not uncommon for the distribution support to be significantly larger than $2^{70}$. Consequently, the existing approaches tend to estimate the desired quantities using a fraction of the required samples, and such estimates are often without the required confidence. The usage of unsound metrics may lead to unsound conclusions, as demonstrated by a recent study where the usage of unsound metric would lead one to conclude that two samplers were indistinguishable (it is worth mentioning that the authors of the study clearly warn the reader about the unsoundness of the underlying metrics)~\cite{DLBS18}.

The researchers in the sub-field of property testing within theoretical computer science have analyzed the sample complexity of testing under different models of samplers and computation. The resulting frameworks have not witnessed widespread adoption to practice due to a lack of samplers that can precisely fit the models under which results are obtained. 
In recent work, Chakraborty and Meel~\cite{CM19}, building on the concepts developed in the condition sampling model (rf. \cite{clement}),  designed the first practical algorithmic procedure, called {\oldbarbarik}, that can rigorously test whether a given sampler samples from the uniform distribution using a constant number of samples, assuming that the given sampler is {\em subquery-consistent} (see Definition~\ref{def:nonadversarial}). Empirically, {\oldbarbarik} was shown to be able to distinguish samplers that were indistinguishable in prior studies based on unsound metrics. While {\oldbarbarik} made significant progress, it is marred by its ability to handle only the uniform distribution. Therefore, one wonders: {\em Can we design an algorithmic framework to test whether the distribution generated by a given sampler is close to a desired (but arbitrary) posterior distribution of interest?}

This paper's primary contribution is the first efficient algorithmic framework, \barbarik,  to test whether the distribution generated by a sampler is $\varepsilon$-close or $\eta$-far from the desired distribution specified by the set of constraints  $\varphi$ and a weight function $\mathtt{wt}$. In contrast to the statistical techniques that require an exponential or sub-exponential number of samples for samplers whose support can be represented by $n$ bits, the number of samples required by {\barbarik} depends on the {\em tilt} of the distribution, where {\em tilt} is defined as the maximum ratio of non-zero weights of two solutions of $\varphi$.  Like {\oldbarbarik}, the key technical idea of {\barbarik} sits at the intersection of {\em property testing} and {\em formal methods} and uses ideas from conditional sampling and employs chain formulas. However, the key algorithmic framework of {\barbarik} differs significantly from {\oldbarbarik}, and, as demonstrated, the proof of its correctness and sample complexity requires an entirely new set of technical arguments.

Given access to an ideal sampler $\mathcal{A}$, {\barbarik} accepts every sampler that is $\varepsilon$-close to $\mathcal{A}$ while its ability to reject a sampler that is $\eta$-far from $\mathcal{A}$ assumes that the sampler under test is \textit{subquery consistent}. 
Since {\barbarik} assumes access to an ideal sampler, one might wonder if a tester such as {\barbarik} is needed when we already have access to an ideal sampler. Since sampling is computationally intractable, it is almost always the case that an ideal sampler $\mathcal{A}$ is quite slow and one would prefer to use some other efficient sampler $\mathcal{G}$ instead of $\mathcal{A}$, if $\mathcal{G}$ can be certified to be close to $\mathcal{A}$. 





To demonstrate the practical efficiency of {\barbarik}, we developed a prototype implementation in Python and performed an experimental evaluation with several samplers. While our framework does not put a restriction on the representation of $\mathtt{wt}$, we perform empirical validation with weight distributions corresponding to log-linear models, a widely used class of distributions.  Our empirical evaluation shows that {\barbarik} returns {\accept} for the samplers with formal guarantees but returns {\reject} for other samplers that are without formal guarantees. Our ability to reject samplers provides evidence in support of our assumption of subquery consistency of samplers. We believe our formalization of testing of samplers and the design of the algorithmic procedure, {\barbarik}, contributes to the design of {\em randomized formal methods} for verified AI, a principle argued by Seshia et al~\cite{SSS16}. 

\section{Notations and Preliminaries}\label{sec:prelims}
A Boolean variable is denoted by a lowercase letter.  For a Boolean formula $\varphi$, the set of variables appearing in $\varphi$, called the \textit{support} of $\varphi$, is denoted by $Supp(\varphi)$. An assignment $\sigma \in \{0,1\}^{|Supp(\varphi)|}$ to the variables of $\varphi$ is a \textit{satisfying assignment} or \textit{witness}  if it makes $\varphi$  evaluate to $1$. We denote the set of all satisfying assignments of $\varphi$ as $R_{\varphi}$. For $S\subseteq Supp(\varphi)$, we use $\sigma_{\downarrow{S}}$ to indicate the projection of $\sigma$ over the set of variables in $S$. And we denote by $R_{\varphi_{\downarrow{S}}}$ the set $\{\sigma_{\downarrow{S}}\ |\ \sigma\in R_{\varphi}\}$.

  \begin{definition}[Weight Function]
 	For a set $S$ of Boolean variables, a weight function $\mathtt{wt}: \{0,1\}^{|S|} \rightarrow (0,1)$ maps each assignment to some weight.
 \end{definition}

\begin{definition}[Sampler]
A sampler $\mathcal{G}(\varphi, S, \mathtt{wt}, \tau)$ is a randomized algorithm that takes in a Boolean formula $\varphi$, a weight function $\mathtt{wt}$, a set $S\subseteq Supp(\varphi)$ and a positive integer $\tau$ and outputs $\tau$  independent samples from $R_{\varphi_{\downarrow {S}}}$.  For brevity of notation we will omit arguments $\varphi, S, \mathtt{wt}, \tau$, whenever may sometimes refer to a sampler as $\mathcal{G}(\varphi)$ or simply, $\mathcal{G}$.

For any $\sigma \in \{0,1\}^{|S|}$ the probability of the sampler $\mathcal{G}$ outputting $\sigma$ is denoted by $p_\mathcal{G}(\varphi, S, \sigma)$ (or $p_{\mathcal{G}}(\varphi, \sigma)$ when the set $S$ in question is clear from the context). 
\end{definition}

We use $D_{\mathcal{G}(\varphi, S)}$  to represent the distribution induced by $\mathcal{G}(\varphi,S)$ on $R_{\varphi_{\downarrow{S}}}$. When the set $S$ is understood from the context we will denote $D_{\mathcal{G}(\varphi, S)}$ by $D_{\mathcal{G}(\varphi)}$.

\begin{definition}[Ideal Sampler]\label{def:ideal_weighted}
For a weight function $\mathtt{wt}$, a sampler $\mathcal{A}(\varphi ,S, \tau)$ is called an ideal sampler w.r.t. weight function $\mathtt{wt}$ if for all $\sigma \in R_{\varphi_{\downarrow{S}}}$:
$	p_\mathcal{A}(\varphi,S, \mathtt{wt}, \sigma) = \frac{\wt{\sigma} }{\sum_{\sigma'\in R_{\varphi_{\downarrow{S}}}}\wt{\sigma'}}$.
In the rest of the paper, $\mathcal{A}(\cdot,\cdot,\cdot, \cdot)$  denotes the ideal sampler. 
When $\mathtt{wt}(\sigma) = \frac{1}{|R_{\varphi}|}$ then the ideal sampler is called a uniform sampler. 
\end{definition}


\begin{definition}[Tilt]\label{defn:tilt}
    For a Boolean formula $\varphi$ and weight function $\mathtt{wt}$, we define 
   $ tilt(\mathtt{wt},\varphi) = 
                 \underset{\sigma_1,\sigma_2 \in R_\varphi}{\max}\;\frac{\mathtt{wt}(\sigma_1)}{\mathtt{wt}( \sigma_2)}$.
\end{definition}
Our goal is to design a program that can test the quality of a sampler with respect to an ideal sampler.
We use two different notions of distance of the sampler from the ideal sampler.

\begin{definition}[$\varepsilon$-closeness and $\eta$-farness]\label{def:closeness_and_farness}
A sampler $\mathcal{G}$ is $\varepsilon$-multiplicative-close (or simply $\varepsilon$-close) to an ideal sampler $\mathcal{A}$, if for all $\varphi$ and  all  $\sigma \in R_{\varphi}$, we have
$ (1-\varepsilon)p_\mathcal{A}(\varphi,\sigma)\leq p_\mathcal{G}(\varphi,\sigma) \leq (1+\varepsilon)p_\mathcal{A}(\varphi,\sigma).$
For a formula $\varphi$, a sampler $\mathcal{G}(\varphi)$ is $\eta$-$\ell_1$-far (or simply $\eta$-far) from the ideal sampler $\mathcal{A}(\varphi)$, if
  $ \sum_{\sigma \in R_{\varphi}} |p_\mathcal{A}(\varphi, \sigma) - p_\mathcal{G}( \varphi, \sigma) | \geq \eta$
\end{definition}

It is worth emphasising that the  asymmetry in the  notions of $\varepsilon$-close and $\eta$-far stems from the availability of practical samplers. Since the available off-the-shelf solvers with theoretical guarantees provide the guarantee of $\varepsilon$-closeness, we are  interested in accepting a sampler that is $\varepsilon$-close~\cite{GSS07,EGSS13a,CMV13a,CFMSV15}. On the other hand, we would like to be more forgiving to the samplers without guarantees and would like to reject only if they are $\eta$-far in $\ell_1$ distance, a notion more relaxed than multiplicative closeness. 

\begin{definition}[$(\varepsilon,\eta, \delta)$-tester for samplers]
A $(\varepsilon, \eta, \delta)$-tester for samplers is a randomized algorithm that takes a sampler $\mathcal{G}$, an ideal sampler $\mathcal{A}$, a tolerance parameter $\varepsilon$, an intolerance parameter $\eta$, a guarantee parameter $\delta$ and a CNF formula $\varphi$ such that (1) If $\mathcal{G}(\varphi)$ is  $\varepsilon$-close to $\mathcal{A}(\varphi)$,  then the tester returns {\accept} with probability at least $(1-\delta)$, and (2) If $\mathcal{G}(\varphi)$ is $\eta$-far from $\mathcal{A}(\varphi)$ then the tester returns {\reject} with probability at least $(1-\delta)$.

\end{definition}

\subsection{Chain Formula}

A crucial component in our algorithm is the chain formula. Chain formulas, introduced in \cite{CFMV15}, are a special class of Boolean formulas. 
Given a positive integer $k$ and $m$, chain formulas provide an efficient construction of a Boolean formula $\psi_{k,m}$ with exactly $k$ satisfying assignments with $\ceil{log(k)} \leq m$ variables. We employ chain formulas for inverse transform sampling and in the subroutine {\fulcrumkernel}.

\begin{definition}~\cite{CFMV15}
	Let $c_1 c_2 \cdots c_m$ be the $m$-bit binary representation of $k$,
	where $c_m$ is the least significant bit.
	We then construct a chain formula $\varphi_{k,m}(\cdot)$
	on $m$   variables $a_1, \ldots a_m$ as follows.  For every $j$
	in $\{1, \ldots m-1\}$, let $C_j$ be the connector ``$\vee$'' if
	$c_j=1$, and the connector ``$\wedge$'' if $c_j=0$.  Define
	\[
	\varphi_{k,m}(a_1,\cdots a_m) = a_1 \,C_1\, (a_2 \,C_2 (\cdots (a_{m-1}
	\,C_{m-1}\, a_m)\cdots))
	\]
\end{definition}

For example, consider $k = 11$ and $m = 4$.  The binary
representation of $11$ using $4$ bits is $1011$. Therefore,
$\varphi_{5,4}(a_1, a_2, a_3, a_4)= a_1\vee(a_2\wedge(a_3\vee a_4))$.

\begin{lemma}~\cite{CFMV15}
	Let $m > 0$ be a natural number, $k < 2^m$ , and $\varphi_{k,m}$
	as defined above. Then
	$|\varphi_{k,m}|$ is linear in $m$ and
	$\varphi_{k,m}$ has exactly $k$ satisfying assignments.
	Every chain formula $\psi$ on $n$ variables is equivalent to a CNF
	formula $\psi^{CNF}$ having at most
	$n$ clauses.  In addition, $|\psi^{CNF}|$  is
	in $O(n^2)$.
\end{lemma}

\subsection{Barbarik2Kernel and the Subquery Consistency Assumption}

{\fulcrumkernel} is a crucial subroutine that we use in our algorithm to help us draw  {\em conditional samples} from $R_{\varphi_{\downarrow {S}}}$. This is similar to the subroutine {\kernel} used by the {\oldbarbarik} in \cite{CM19}. We will now define a collection of functions {\kernelfamily}. 

\begin{definition}\label{def:kernel}
	{\kernelfamily} is family of functions that take a Boolean formula $\varphi$, a set of variables $S \subseteq Supp(\varphi)$, and two assignments $\sigma_1, \sigma_2 \in R_{\varphi \downarrow S}$,  and return $\hat{\varphi}$ such that $R_{\hat{\varphi}\downarrow{S}} = \{\sigma_1, \sigma_2 \}$. 
\end{definition}

\cite{CM19}  introduced the notion of {\em non-adversarial assumption}, which was crucial in their analysis. We rename the notion of {\em subquery consistency} to better capture its intended properties, defined below.

\begin{definition}\label{def:nonadversarial} 
	Let ${\fulcrumkernel}\in{ \kernelfamily}$. A sampler $\mathcal{G}$ is subquery consistent w.r.t.  a particular {\fulcrumkernel} for $\varphi$ if for all $S \subseteq Supp(\varphi)$, $  \sigma_1, \sigma_2 \in  R_{\varphi_{\downarrow {S}}}$, let   $\hat{\varphi} \gets \fulcrumkernel(\varphi, S, \sigma_1, \sigma_2)$ then
	the output of $\mathcal{G}(\hat{\varphi}, \mathtt{wt}, S, \tau)$ is $\tau$ independent samples from the conditional distribution $\mathcal{D}_{\mathcal{G}(\varphi) |{T}}$, where $T=\{\sigma_1, \sigma_2\}$.
\end{definition}

Similar to the usage of {\em non-adversarial assumption} in the correctness analysis of {\oldbarbarik}~\cite{CM19}, the notion of {\em subquery consistency} would play a crucial role in our analysis.  Since each subquery can be viewed as conditioning and  given that conditioning is a fundamental operation, one would expect that off the shelf samplers would be subquery consistent. At the same time, in contrast to practical applications, the set $T$ is arbitrarily chosen, and therefore, it is possible that certain samplers do not satisfy the property of subquery consistency. It is, however, not known how to test whether a sampler is  subquery consistent w.r.t  a particular {\fulcrumkernel}. 
While our empirical evaluation provides weak evidence to our claim that off the shelf samplers are subquery consistent, we believe  checking whether a sampler is subquery consistent is an interesting and important problem for future work.

\section{Related Work}\label{sec:related}
Distribution testing involves testing whether an unknown probability distribution is identical or close to  a  given distribution. This problem has been studied extensively in the property testing literature~\cite{conditional1, CanonneRS15,ValiantV11,Valiant08} . The sample space is exponential, and for many fundamental distributions, including uniform, it is prohibitively expensive in terms of samples to verify closeness. This led to the development of the conditional sampling model \cite{conditional1, CanonneRS15}, which can provide sub-linear or even \textit{constant} sample complexities for the testing of the above-given properties\cite{clement, KamathT19,  BhattacharyyaC18, kane17, Jayaram20}. A detailed discussion on prior work in property testing and their  relationship to {\barbarik} is given in Appendix \ref{sec:propertytestingoverview}. 



The first practically efficient algorithm for  verification of samplers with a formal proof of correctness  was presented by Chakraborty and Meel in form of {\oldbarbarik}~\cite{CM19}. 
The central idea of {\oldbarbarik}, building on the work of Chakraborty et al.~\cite{conditional1} and Canonne et al.~\cite{CanonneRS15}, was that if one can have conditional samples from the distribution, then one can test properties of the distribution using only a constant number of conditional samples.

{\oldbarbarik} constructs a two-element set $T \subset R_{\varphi}$, with one element drawn according to the distribution $D_{\mathcal{G}(\varphi)}$ and one element drawn uniformly at random from the set $R_{\varphi}$. Using a subroutine {\kernel} Chakraborty et al.~argued that one can draw samples from the conditional distribution $D_{\mathcal{G}(\varphi)\mid T}$.  Their sample complexity was $\Tilde{O}(1/(\eta -2\varepsilon)^4)$.  They proved that if a sampler $\mathcal{G}$ is $\varepsilon$-close to a uniform sampler then {\oldbarbarik} will accept with probability at least $(1-\delta)$, while if $\mathcal{G}(\varphi)$ is $\eta$-far from the uniform sampler and if $\mathcal{G}$ is subquery consistent w.r.t {\kernel} for $\varphi$, then {\oldbarbarik} rejects with probability at least $(1-\delta)$. Their underlying assumption was that many samplers that are in use would in fact be {\em subquery consistent} and the success of {\oldbarbarik} in rejecting several samplers provides evidence in support of the aforementioned assumption. 
They used {\oldbarbarik} to test the correctness of samplers like {\SearchTreeSampler}, {\QuickSampler}, and {\WeightGen}. 

Note that {\oldbarbarik} can only distinguish a uniform sampler from a far-from uniform sampler, and the techniques used cannot be generalized easily to the case where the ideal sampler is not necessarily uniform. While {\barbarik}, that we present in this paper, does borrow several techniques from {\oldbarbarik}, including drawing inspiration from the concept of conditional sampling for their design; {\barbarik} is very different from {\oldbarbarik} both in terms of the algorithmic design and its implementation.  

\section{An overview of the {\barbarik} Algorithm}\label{sec:barbarik}
In this section, we present the algorithmic framework of {\barbarik}, the pseudocode, presented as Algorithm~\ref{algo:theory_new}, and then the theoretical justification for the algorithm.
{\barbarik} takes as input a blackbox sampler $\mathcal{G}$, a Boolean formula $\varphi$ with the associated weight function $\mathtt{wt}$ and three parameters $(\varepsilon,\eta,\delta)$. It also has access to an ideal sampler $\mathcal{A}$. 
{\barbarik} is an $(\varepsilon,\eta,\delta)$-tester for samplers. Also if {\barbarik} returns {\reject} (that is, when $\mathcal{G}$ is $\eta$-far from $\mathcal{A}$), it provides as witness a new formula $\hat\varphi$ which is similar to $\varphi$, except that $\hat\varphi$ has only two assignments to the variables in $S$ (namely $\sigma_1$ and $\sigma_2$) that can be extended to satisfying assignments of $\hat{\varphi}$ and the relative probability masses of $\sigma_1$ and $\sigma_2$ in $\mathcal{D}_{\mathcal{G}}$ are significantly different from that in $\mathcal{D}_{\mathcal{A}}$.


The core idea of {\barbarik} is that for verifying the quality of the sampler ${\mathcal{G}(\varphi)}$, we can proceed in two stages. In the first stage, if the sampler is far from the ideal sampler $\mathcal{A}$, we hope to find a witness (in the form of two satisfying assignments) for farness with good probability. This can be guaranteed by drawing one sample each from $\mathcal{D}_{\mathcal{G}(\varphi)}$ and $\mathcal{D}_{\mathcal{A}(\varphi)}$.
In the second stage, we confirm whether the witness is indeed far. That is, if the witness is the ($\sigma_1, \sigma_2$) pair, we check that the probability of $\sigma_1$ and $\sigma_2$  in $\mathcal{D}_{\mathcal{G}(\varphi)}$ and $\mathcal{D}_{\mathcal{A}(\varphi)}$ are similar or not.

Here {\barbarik} differs from {\oldbarbarik} in a significant way. {\oldbarbarik} employs a bucketing strategy. But, {\barbarik} chooses a simpler yet equally effective method to check the similarity between $\sigma_1$ and $\sigma_2$.
This is also the most difficult stage of the tester as one may have to draw a exponential number of samples to confirm the similarity. We manage this by drawing samples from the conditional distribution $\mathcal{D}_{\mathcal{G}(\varphi)\mid\{\sigma_1, \sigma_2\}}$ instead of $\mathcal{D}_{\mathcal{G}(\varphi)}$. Since $\mathcal{D}_{\mathcal{G}(\varphi)\mid\{\sigma_1, \sigma_2\}}$ is supported on a set of size only two estimating the distance of $\mathcal{D}_{\mathcal{G}(\varphi)\mid\{\sigma_1, \sigma_2\}}$ from $\mathcal{D}_{\mathcal{A}(\varphi)\mid\{\sigma_1, \sigma_2\}}$ can be done with constant number of samples.

Now since we do not have direct access to the distribution $\mathcal{D}_{\mathcal{G}(\varphi)\mid\{\sigma_1, \sigma_2\}}$ we circumvent the problem by drawing samples from a new distribution  $\mathcal{D}_{\mathcal{G}(\hat{\varphi})}$ where $\hat{\varphi}$ is obtained from $\varphi$ and has similar structure as $\varphi$ (with $Supp(\varphi)\subseteq Supp(\hat{\varphi})$) and there are only two assignments (namely $\sigma_1$ and $\sigma_2$) to the variables in $Supp(\varphi)$ that can be extended to the satisfying assignments of $\hat{\varphi}$. The subroutine \fulcrumkernel~helps us simulate the drawing of samples from $\mathcal{D}_{\mathcal{G}(\varphi)\mid\{\sigma_1, \sigma_2\}}$ by drawing samples from  $\mathcal{D}_{\mathcal{G}(\hat{\varphi})}$. The subroutine {\bias}~helps to estimate the distance of
$\mathcal{D}_{\mathcal{G}(\hat{\varphi})}$ from  $\mathcal{D}_{\mathcal{A}(\hat{\varphi})}$.

Finally, we repeat the whole process for a certain number of rounds, and we argue that if the sampler is indeed far then, with high probability, in at least one round, we will find a witness of farness and confirm that the witness is indeed far. On the other hand, if the sampler is close to ideal, then there does not exist any such witness of farness.

{\barbarik} accesses two subroutines, {\bias} and {\fulcrumkernel}:
{\bias}$(\hat{\sigma},\Gamma,S)$ takes as input an assignment $\hat\sigma$, a list $\Gamma$ of assignments and a sampling set $S$. It returns the fraction of assignments of $\Gamma$ whose projections on $S$ is equal to $\hat{\sigma}$. 
{\fulcrumkernel}$(\varphi, \sigma_1, \sigma_2)$ is a  {\fulcrumkernel} subroutine (Definition \ref{def:kernel}).
Its aim is to create a $\hat\varphi$ such the behaviour of the sampler on $\hat\varphi$ is similar to it's behaviour on $\varphi$, i.e. $\mathcal{D}_{\mathcal{G}(\varphi)\mid\{\sigma_1, \sigma_2\}} \approx \mathcal{D}_{\mathcal{G}(\hat\varphi)}$.


In {\barbarik}, in the for loop (in lines \ref{line:barloop2}$-$\ref{line:barreject}), in
each round, the algorithm draws one sample $\sigma_1$ according to the
distribution $\mathcal{D}_{\mathcal{G}(\varphi)}$ and one sample $\sigma_2$
according to the ideal distribution on $R_{\varphi}$ (line \ref{line:one_sample}). 
In the case that $\sigma_1 = \sigma_2$ it moves the to next iteration (in line \ref{line:ifequal1}-\ref{line:ifequal2}).
In line \ref{line:barkernel}, the
subroutine {\fulcrumkernel} uses $\varphi$, the two samples $\sigma_1$ and $\sigma_2$, to output a
new formula $\hat{\varphi}$ such that
$Supp(\varphi)\subseteq\ Supp(\hat{\varphi})$.
On line \ref{line:barsample3}, {\barbarik} draws a list, $\Gamma_3$, of $N$ samples according to the
distribution $\mathcal{D}_{\mathcal{G}(\hat{\varphi})}$. {\fulcrumkernel} ensures that
for all $\sigma\in \Gamma_3$, $\sigma_{\downarrow S}$ is either $\sigma_1$ or
$\sigma_2$. In line \ref{line:barbias1} {\barbarik} uses {\bias} to compute the fraction of samples that are equal to $\sigma_1$ (on the variable set $S$), and if the fraction is greater than the threshold then {\barbarik} returns {\reject} (in line \ref{line:barreject}).


\begin{figure}
	\begin{minipage}[t]{0.50\textwidth}
		\begin{algorithm}[H]
			\footnotesize
			\caption{{\barbarik}($\mathcal{G},\mathcal{A}, \varepsilon,\eta,\delta,\varphi,S,\mathtt{wt}$)}
			\label{algo:theory_new}
			\begin{algorithmic}[1]
				\State $t \gets  ln(1/\delta) ln\left(\frac{10}{10-\eta(\eta-6\varepsilon)}\right)^{-1} $ \label{line:trials}
				\State $ n \gets  8  ln\left(t/\delta\right)$
				\State $lo = (1+\varepsilon)/(1-\varepsilon)$
				\State $hi = 1 + (\eta+6\varepsilon)/4$
				\State $\Gamma_1 \leftarrow \mathcal{G}(\varphi, S, t)$;\label{line:barsample1}
				\State $\Gamma_2 \gets \mathcal{A}(\varphi, S, t)$;\label{line:barsample2}
				\For{$i = 1$ to $t$} \label{line:barloop2}
				\State $\sigma_1 \gets \Gamma_1[i]$; $\sigma_2 \gets \Gamma_2[i];$\label{line:one_sample}
				\If{$\sigma_1 = \sigma_2$} \label{line:ifequal1}
				\State $\mathbf{continue}$\label{line:ifequal2}
				\EndIf
				\State $\alpha \gets \mathtt{wt}(\sigma_1)/\mathtt{wt}(\sigma_2)$\label{line:alpha}
				\State $L \gets \left(\alpha \cdot lo\right)/ \left(1+\alpha \cdot lo\right)$
				\State $H \gets \left(\alpha \cdot hi\right)/\left(1+\alpha \cdot hi\right)$
				\State $T = (H + L)/2$
				\State $N \gets n \cdot  H/(H - L)^2$ \label{line:ncalc}
				\State $\hat{\varphi} \leftarrow \fulcrumkernel(\varphi, \sigma_1, \sigma_2)$ \label{line:barkernel}
				\State $\Gamma_3 \gets \mathcal{G}(\hat{\varphi}, S, N)$  \label{line:barsample3}        
				\State $Bias\gets \bias(\sigma_1, \Gamma_3, S) \label{line:barbias1}$
				\If{$Bias > T$} \label{line:compare}
				\State \Return {\reject}  \label{line:barreject}
				\EndIf
				\EndFor
				\State \Return {\accept} \label{line:baraccept}
			\end{algorithmic}
		\end{algorithm}
	\end{minipage}
	\hfill
	\begin{minipage}[t]{0.48\textwidth}
		\begin{algorithm}[H]
			\caption{$\fulcrumkernel (\varphi, \sigma_1, \sigma_2)$}\label{alg:kernel}
			\begin{algorithmic}[1]
				\State $m \gets 12, k \gets 2^m-1$
				\State $\mathbf{Lits_1} \gets (\sigma_1 \setminus \sigma_2)$ \label{line:kernel-lits-one}
				\State $\mathbf{Lits_2} \gets (\sigma_2 \setminus \sigma_1)$ \label{line:kernel-lits-two}
				\State $\mathbf{V} \gets  NewVars(\varphi, m); $
				\State $\hat{\varphi} \gets \varphi \wedge (\sigma_1 \vee \sigma_2)$ \label{line:kernel-construct-begin}
				\State{$l \sim \mathbf{Lits_1} \cup \mathbf{Lits_2}$} \label{line:kernel-loop-begin}
				\State $\hat{\varphi} \gets \hat{\varphi} \wedge (\lnot l \rightarrow \psi_{k,m}
				(\mathbf{V}))$
				\State $\hat{\varphi} \gets \hat{\varphi} \wedge (l \rightarrow \psi_{k,m}
				(\mathbf{V}))$            \label{line:kernel-loop-end}
				\State \Return $\hat{\varphi}$
			\end{algorithmic}
		\end{algorithm}
		
		\begin{algorithm}[H]
			\caption{\bias($\hat{\sigma}$, $\Gamma$, $S$)}\label{alg:bias}
			\begin{algorithmic}[1]
				\State $count$ = 0
				\For {$\sigma \in \Gamma$}
				\If {$\sigma_{\downarrow S} = \hat{\sigma}$}
				\State count $\gets$ count +1
				\EndIf
				\EndFor
				\State \Return $\frac{count}{|\Gamma|}$        
			\end{algorithmic}
		\end{algorithm}
	\end{minipage}
\end{figure}

Algorithm~\ref{alg:kernel} presents the pseudocode of subroutine {\fulcrumkernel}. As stated above, {\fulcrumkernel} takes in a Boolean formula $\varphi$, a set $S\subseteq Supp(\varphi)$ and two partial assignments $\sigma_1, \sigma_2 \in R_{\varphi \downarrow S}$ .  Since the set $S$ is implicit from $\sigma_1$ and $\sigma_2$ it may not be explicitly given as an  input. {\fulcrumkernel} assumes access to subroutine $NewVars$ which takes in two parameters, a formula $\varphi$ and a number $m$, and returns a set  of $m$ fresh variables that do not appear in $\varphi$. {\fulcrumkernel} first constructs two sets of literals, denoted by $\mathbf{Lits_1}$ (resp. $\mathbf{Lits_2}$), which appear in $\sigma_1$ (resp. $\sigma_2$) but not $\sigma_2$ (resp. $\sigma_1$).
The algorithm then constructs the formula $\hat\varphi$. First it generates $\varphi \wedge (\sigma_1 \vee \sigma_2)$ on Line \ref{line:kernel-construct-begin}, a formula with exactly two solutions. Next, it randomly chooses a literal $l$ from $\mathbf{Lits_1} \cup \mathbf{Lits_2}$ and constructs a chain formula $(l \rightarrow \psi_{k,m})$ over the fresh Boolean variables ${\mathbf{V}[1], \mathbf{V}[2]\cdots, \mathbf{V}[m]}$ where $k$ is the number of satisfying assignments the formula has. Conjuncting the two generated formulas, we get $\hat\varphi \equiv \varphi \wedge (\sigma_1 \vee \sigma_2) $.
Therefore, at the end of \fulcrumkernel, i.e. line~\ref{line:kernel-loop-end}, $\hat{\varphi}$ has $2k$ solutions.  We choose the value of $k$ such that it is odd (see~\cite{CFMV15}).  The chain formula is linked to a random Boolean literal from the given set of literals for two reasons,
\begin{enumerate}
    \item An ideal or $\varepsilon$-close to ideal sampler would not be affected by the randomization and would generate the same distribution over $\hat{\varphi}$ as it does over ${\varphi \wedge (\sigma_1 \vee \sigma_2)}$.
	\item If the sampler under test $\mathcal{G}$ is $\eta$-far from ideal, then we want to construct a formula which \textit{cannot} be easily guessed by $\mathcal{A}$. We wish to avoid the scenario where $\mathcal{A}$, an $\eta$-far sampler on $\varphi$, somehow behaves as an almost-ideal sampler over $\hat\varphi$ and hence manages to fool {\barbarik}.
	
\end{enumerate}

\subsection{Theoretical Analysis}
The following theorem gives the mathematical guarantee about the correctness of {\barbarik}.

\begin{theorem} \label{thm:main}
	Given sampler $\mathcal{G}$, ideal sampler $\mathcal{A}$, $\varepsilon < \frac{1}{3}$, $\eta > 6\varepsilon$, $\delta$, $\varphi$  and weight function $\mathtt{wt}$, {\barbarik} needs at most
	$\widetilde{O}\left(\frac{tilt(\mathtt{wt},\varphi)^2}{\eta(\eta - 6\varepsilon)^3}\right)$ samples, where $\widetilde{O}$ hides a poly logarithmic factor of $1/\delta$. 
	\begin{itemize}
		\item  If $\mathcal{G}$ is an $\varepsilon$-close  to $\mathcal{A}$ then  {\barbarik} returns {\accept} with probability at least $(1-\delta)$.
		\item If $\mathcal{G}$ is subquery consistent w.r.t {\fulcrumkernel} and if the distribution $\mathcal{D}_{\mathcal{G}(\varphi)}$ is $\eta$-far from the ideal sampler  then {\barbarik}
		returns {\reject} with  probability at least $(1-\delta)$.
	\end{itemize}
\end{theorem}





Note that if  $\mathcal{G}$ is $\varepsilon$-close to  $\mathcal{A}$ then {\barbarik}  accepts (with high probability) even if the sampler $\mathcal{G}$ isn't subquery consistent w.r.t  {\fulcrumkernel}. 
It is also worth noting that {\barbarik} terminates with {\reject} as soon as the check in line~\ref{line:compare} succeeds. Therefore, we expect {\barbarik} to require significantly less number of samples when it returns {\reject}. Furthermore, in the case of {\accept}, the bound on $N$, as calculated on line~\ref{line:ncalc} in terms of $tilt$, is pessimistic as the probability of observing $\sigma_1$ and $\sigma_2$ such that $\alpha \approx tilt$ for a sampler close to ideal is very small when the tilt is large. The proof of Theorem~\ref{thm:main} is presented in Appendix~\ref{sec:proof}.


\section{Evaluation}
\label{sec:experiment}
The objective of our evaluation was to answer the following questions:
\begin{enumerate}
	\item[\bf{RQ1.}] Is {\barbarik} able to distinguish between off-the-shelf samplers  by returning {\accept} for samplers $\varepsilon$-close to the ideal distribution and {\reject} for the $\eta$-far samplers?
	\item[\bf{RQ2.}] What improvements do we observe over the baseline?  
	\item[\bf{RQ3.}] How does the required number of samples scale with the $tilt(\mathtt{wt}, \varphi)$ of the distribution?
\end{enumerate}
To evaluate the runtime performance of {\barbarik} and test the
quality of some state of the art samplers, we implemented
a prototype of {\barbarik} in Python. Our algorithm utilizes an ideal sampler, for which we use the state of the art sampler WAPS~\citep{GSRM19}.
All experiments were conducted on a high performance computing cluster with 600 E5-2690 v3 $@ 2.60$GHz CPU cores. For each benchmark, we use a single core with a timeout of 24 hours.  The detailed logs along with list of benchmarks and the runtime code employed to run the experiments are available at \url{http://doi.org/10.5281/zenodo.4107136}.

We focus on the log-linear distributions given their ubiquity of usage in machine learning; a formal description is provided in Appendix~\ref{sec:weighted_dist} for completeness. Observe  that {\barbarik} does not put any restrictions on the representation of the weight distribution. We conducted our experiments on 72 publicly available benchmarks, which have been employed in the evaluation of samplers proposed in the past~\cite{CFMSV14,DLBS18}. The $tilt$ of the benchmarks spans many orders of magnitude, between $1$ and $10^{11}$.

\paragraph*{Samplers Tested}
The past few years have witnessed a multitude of sampling techniques ranging from variational methods~\cite{WJ2008}, MCMC-based techniques~\cite{JS96,M02}, mutation-based sampling~\cite{DLBS18}, importance sampling-based methods~\cite{EGS12}, knowledge-compilation techniques~\cite{GSRM19} and the like. The conceptual simplicity of uniform samplers encourages designers to tune their algorithms for uniform sampling, and the standard technique for weighted sampling employs the well-known method of the inverse transform. For the sake of completeness, we provide a detailed discussion of the transformation technique in Appendix~\ref{sec:weighted_dist}

We perform empirical evaluation with the three state of the art samplers {\wWeightGen}, {\wQuickSampler}, and {\wSearchTreeSampler} constructed by augmenting inverse sampling  with underlying samplers {\WeightGen}~\cite{CFMSV14}, {\QuickSampler}~\cite{DLBS18} and {\SearchTreeSampler}(SearchTreeSampler)~\cite{EGS12} respectively. 

While {\wWeightGen} is known to have theoretical guarantees of $\varepsilon-$closeness, there is no theoretical analysis of the distributions generated by {\wQuickSampler} and {\wSearchTreeSampler}.  Of the 72 instances, {\wWeightGen} can handle only  35 instances while {\wQuickSampler} and {\wSearchTreeSampler} can handle all the 72  instances. The variation in the number of instances that are amenable to sampling for a particular sampler highlights the trade-off between the runtime performance and theoretical guarantees.  It is perhaps worth emphasizing that {\wQuickSampler} and {\wSearchTreeSampler} are significantly more efficient in runtime performance than the ideal sampler WAPS. 

\paragraph*{Test Parameters}
We set tolerance parameter $\varepsilon$, intolerance parameter $\eta$, and confidence $\delta$ for {\barbarik} to be and 0.1, 1.6 and 0.2 respectively. The chosen setting of parameters implies that for a given Boolean formula $\varphi$, if the sampler under test $\mathcal{G}(\varphi)$ is $\varepsilon$-close to the ideal sampler, then {\barbarik} returns {\accept} with probability at least 0.8, otherwise if the sampler is $\eta$-far from ideal sampler then {\barbarik} returns {\reject} with probability at least 0.8.  Note that, the number of samples required for {\accept} depends only on the parameters ($\varepsilon, \eta,\delta$) and $tilt(\mathtt{wt}, \varphi)$.  We instantiate {\fulcrumkernel} with the values $m = 12$ and $k = 2^{m}-1$. Observe that Theorem~\ref{thm:main} does not put restrictions on $k$ and $m$. 

\paragraph{Description of the table} We present the experimental results in Table \ref{table:results}. Due to lack of space, we present results for a subset of benchmarks while the extended table is available in the supplementary material. The first column indicates the name of the benchmark, the second the $tilt$, and the following columns indicate the outcome of the experiments with {\wWeightGen}, {\wSearchTreeSampler} and {\wQuickSampler} in that order. Every cell in the table has two entries. In the second column, the first entry shows the value of $tilt$ for the corresponding benchmark, while in the other columns, it contains ``A'' and ``R'' to indicate the output of {\barbarik} for the corresponding sampler. The second entry for the cells in the column corresponding to $tilt$ indicates the theoretical upper bound on the samples required for {\barbarik} to terminate, while for rest of the columns, the second entry indicates the number of samples consumed by {\barbarik} for the corresponding instance and the sampler. 
\begin{table}
\footnotesize
\begin{tabular}{{l>{\centering}m{0.16\linewidth}>{\centering\arraybackslash}m{0.18\linewidth}>{\centering\arraybackslash}m{0.16\linewidth}>{\centering\arraybackslash}m{0.18\linewidth}}}
  \toprule  && \multicolumn{3}{c}{\barbarik}  \\ \cmidrule(l){3-5} 
 Benchmark&\shortstack{$tilt$\\ (maxSamp)}&\shortstack{\wWeightGen\\ (samples)}&\shortstack{\wSearchTreeSampler\\ (samples)}&\shortstack{\wQuickSampler\\ (samples)}\\
	\midrule 
	s349\_3\_2&\shortstack{28\\ (3e+07)}&\shortstack{A\\ (1e+05)}&\shortstack{A\\ (1e+05)}&\shortstack{R\\ (22854)}\\ \midrule 
	s820a\_3\_2&\shortstack{37\\ (5e+07)}&\shortstack{A\\ (96212)}&\shortstack{R\\ (87997)}&\shortstack{A\\ (2e+05)}\\ \midrule 
UserServiceImpl.sk&\shortstack{140\\ (6e+08)}&\shortstack{A\\ (1e+05)}&\shortstack{R\\ (1e+05)}&\shortstack{R\\ (4393)}\\ \midrule 
LoginService2.sk&\shortstack{232\\ (2e+09)}&\shortstack{A\\ (1e+05)}&\shortstack{R\\ (38044)}&\shortstack{R\\ (13350)}\\ \midrule 
s349\_7\_4&\shortstack{603\\ (1e+10)}&\shortstack{A\\ (75555)}&\shortstack{R\\ (4284)}&\shortstack{R\\ (5150)}\\ \midrule 
s344\_3\_2&\shortstack{3300\\ (3e+11)}&\shortstack{A\\ (1e+05)}&\shortstack{R\\ (59952)}&\shortstack{R\\ (5150)}\\ \midrule 
s420\_new\_7\_4&\shortstack{3549\\ (4e+11)}&\shortstack{A\\ (82312)}&\shortstack{A\\ (96659)}&\shortstack{R\\ (49955)}\\ \midrule 
54.sk\_12\_97&\shortstack{4e+11\\(6e+27)}&DNS&\shortstack{R\\ (14012)}&\shortstack{R\\ (4627)}\\ \midrule 
s641\_7\_4&\shortstack{9e+07\\(3e+20)}&DNS&\shortstack{R\\ (8747)}&\shortstack{A\\ (1e+06)}\\ \midrule 
s838\_3\_2&\shortstack{2e+08\\(1e+21)}&DNS&\shortstack{R\\ (9504)}&\shortstack{R\\ (4627)}\\  \hline 
		\end{tabular}
\caption{``A"(resp. ``R") represents {\barbarik} returning {\accept}(resp. {\reject}). maxSamp represents the upper bound on the number of samples required by {\barbarik} to return {\accept}/{\reject}.} \label{table:results}
\end{table}

\paragraph{\bf{RQ1}} Our experiments demonstrate that {\barbarik} returns {\reject} for {\wQuickSampler} on 68 benchmarks and {\accept} on the remaining four benchmarks. For {\wSearchTreeSampler} we found {\barbarik} returned {\reject} on 62 of the benchmarks and {\accept} on 7 while it times out on the remaining 3. Since {\wSearchTreeSampler} and {\wQuickSampler} are samplers with no formal guarantees and therefore one may expect them to generation distributions away from the ideal distributions. In this context, the results in Table~\ref{table:results} provide strong evidence for the reasonableness of the {\em subquery consistency} assumption in practice.  

In contrast, {\barbarik} returned {\accept} for {\wWeightGen} on all the 35 benchmarks for which {\wWeightGen} could sample. Recall, {\wWeightGen} formally guarantees $\varepsilon$-closeness of the samples to the required distribution, hence {\barbarik} returning {\accept}  on all the benchmarks provides evidence in support of soundness of {\barbarik}. 

\paragraph{\bf{RQ2}} We also computed the number of samples required by the baseline approach owing to~\cite{batu}. Since the number of samples is so large that exhaustive experimentation is infeasible, we had to resort to estimating the average time taken by a sampler for a given instance. Based on the estimated time, we can estimate the time taken by the baseline for our benchmark set.  We observe that the time taken by the baseline would be over $10^6$ seconds for 43, 42 and 16 benchmarks for {\wQuickSampler}, {\wSearchTreeSampler} and {\wWeightGen} respectively. In this context, it is worth highlighting that {\barbarik} terminates within 24 hours for all the instances for all the samplers. We observe that the geometric means of the speedups over the baseline approach are $10^{5.0}, 10^{20.2}$ and $58$ for {\wSearchTreeSampler}, {\wQuickSampler} and, {\wWeightGen} respectively. The lower speedup in the case of {\wWeightGen} owes to its ability to handle only small benchmarks, for which the number of models was not very large. The extended results are available in Appendix~\ref{sec:detailedresults}. 
 
\paragraph{\bf{RQ3}} The number of trials required (indicated by the the variable $t$ as on Line~\ref{line:barloop2} of Algorithm~\ref{algo:theory_new}) depends only on ($\varepsilon, \eta,\delta$), so for the values we use, $(0.1, 1.6, 0.2)$, we find that we require $t = 14$ trials. The analysis of the algorithm reveals an upper bound on the sample complexity of the tester (See Section~\ref{sec:barbarik}, Theorem~\ref{thm:main}) which is quadratic in terms of the $tilt(\mathtt{wt}, \varphi)$. We now return to Table~\ref{table:results} and observe that the number of samples required by {\barbarik} before returning {\accept} were significantly lower than the theoretical bound provided in the second column. Furthermore, as noted earlier, the number of samples required before {\barbarik} returns {\reject} is typically significantly less than the worst case -- a trend demonstrated in Table~\ref{table:results}.

\section{Conclusion}\label{sec:conclusion}
	In this paper, we study the problem of verifying whether a probabilistic sampler samples from a given discrete distribution. 
	Existing approaches require samples linear in the size of the sampling set, which is commonly exponentially large. We present a conditional sampling technique that can verify the sampler in sample complexity constant in terms of the sampling set. We also test a prototype implementation of our algorithm against three state-of-the-art samplers. 
	
	We noticed that the analytical upper bound on the sample complexity is significantly weak compared to our observed values; this suggests that the bounds could be further tightened. Our algorithm can only deal with those discrete distributions for which the relative probabilities of any two points is easily computable. 
	Since our algorithm does not deal with all possible discrete distributions, extending the approach to other distributions would enable the testing of a broader set of samplers. 

\section*{Broader Impact}

The recent advances in machine learning techniques have led to increased adoption of the said techniques in safety-critical domains. The usage of a technique in a safety-critical domain necessitates appropriate verification methodology. This paper seeks to take a step in this direction and focused on one core component. Our analysis is probabilistic, and therefore, practical adoption of such techniques requires careful design of frameworks to handle failures. 
\section*{Acknowledgements}
We are grateful to Teodora Baluta and Arijit Shaw for the technical help and for the useful comments on the earlier drafts of the work. We are grateful to the anonymous reviewers for their constructive feedback that has greatly improved the quality of the paper. This work was supported in part by the National Research Foundation Singapore under its NRF Fellowship Programme [NRF-NRFFAI1-2019-0004] and the AI Singapore Programme [AISG-RP-2018-005],  and NUS ODPRT Grant [R-252-000-685-13]. The computational work for this article was performed on resources of the National Supercomputing Centre, Singapore \href{https://www.nscc.sg}{https://www.nscc.sg}. Any opinions, findings and conclusions or recommendations expressed in this material are those of the author(s) and do not reflect the views of National Research Foundation, Singapore.

	\bibliographystyle{plainnat}
	\bibliography{Report}

\begin{thebibliography}{38}
\providecommand{\natexlab}[1]{#1}
\providecommand{\url}[1]{\texttt{#1}}
\expandafter\ifx\csname urlstyle\endcsname\relax
  \providecommand{\doi}[1]{doi: #1}\else
  \providecommand{\doi}{doi: \begingroup \urlstyle{rm}\Url}\fi

\bibitem[Acharya et~al.(2018)Acharya, Canonne, and Kamath]{clement}
Jayadev Acharya, Cl{\'{e}}ment~L. Canonne, and Gautam Kamath.
\newblock A chasm between identity and equivalence testing with conditional
  queries.
\newblock \emph{Theory of Computing}, 2018.

\bibitem[Altmann and Sauer(2017)]{AS17}
J{\"u}rgen Altmann and Frank Sauer.
\newblock Autonomous weapon systems and strategic stability.
\newblock \emph{Survival}, 2017.

\bibitem[Andrieu et~al.(2003)Andrieu, De~Freitas, Doucet, and Jordan]{ADDJ03}
Christophe Andrieu, Nando De~Freitas, Arnaud Doucet, and Michael~I Jordan.
\newblock An introduction to {MCMC} for machine learning.
\newblock \emph{Machine learning}, 2003.

\bibitem[Batu et~al.(2013)Batu, Fortnow, Rubinfeld, Smith, and White]{batu}
Tugkan Batu, Lance Fortnow, Ronitt Rubinfeld, Warren~D. Smith, and Patrick
  White.
\newblock Testing closeness of discrete distributions.
\newblock \emph{J. {ACM}}, 2013.

\bibitem[Bhattacharyya and Chakraborty(2018)]{BhattacharyyaC18}
Rishiraj Bhattacharyya and Sourav Chakraborty.
\newblock Property testing of joint distributions using conditional samples.
\newblock \emph{{TOCT}}, 2018.

\bibitem[Brooks and Gelman(1998)]{BG98}
Stephen~P Brooks and Andrew Gelman.
\newblock General methods for monitoring convergence of iterative simulations.
\newblock \emph{Journal of computational and graphical statistics}, 1998.

\bibitem[Brooks et~al.(2011)Brooks, Gelman, Jones, and Meng]{BGJM11}
Steve Brooks, Andrew Gelman, Galin Jones, and Xiao-Li Meng.
\newblock \emph{Handbook of Markov Chain Monte Carlo}.
\newblock CRC press, 2011.

\bibitem[Canonne et~al.(2015)Canonne, Ron, and Servedio]{CanonneRS15}
Cl{\'{e}}ment~L. Canonne, Dana Ron, and Rocco~A. Servedio.
\newblock Testing probability distributions using conditional samples.
\newblock \emph{{SIAM} J. Comput.}, 2015.

\bibitem[Canonne et~al.(2017)Canonne, Diakonikolas, Kane, and Stewart]{kane17}
Cl{\'{e}}ment~L. Canonne, Ilias Diakonikolas, Daniel~M. Kane, and Alistair
  Stewart.
\newblock Testing conditional independence of discrete distributions.
\newblock \emph{CoRR}, 2017.

\bibitem[Chakraborty and Meel(2019)]{CM19}
Sourav Chakraborty and Kuldeep~S. Meel.
\newblock On testing of uniform samplers.
\newblock In \emph{Proc. of AAAI}, 2019.

\bibitem[Chakraborty et~al.(2016)Chakraborty, Fischer, Goldhirsh, and
  Matsliah]{conditional1}
Sourav Chakraborty, Eldar Fischer, Yonatan Goldhirsh, and Arie Matsliah.
\newblock On the power of conditional samples in distribution testing.
\newblock \emph{{SIAM} J. Comput.}, 2016.

\bibitem[Chakraborty et~al.(2013)Chakraborty, Meel, and Vardi]{CMV13a}
Supratik Chakraborty, Kuldeep~S. Meel, and Moshe~Y. Vardi.
\newblock {A Scalable and Nearly Uniform Generator of {SAT} Witnesses}.
\newblock In \emph{Proc. of CAV}, 2013.

\bibitem[Chakraborty et~al.(2014)Chakraborty, Fremont, Meel, Seshia, and
  Vardi]{CFMSV14}
Supratik Chakraborty, Daniel~J. Fremont, Kuldeep~S. Meel, Sanjit~A. Seshia, and
  Moshe~Y. Vardi.
\newblock Distribution-aware sampling and weighted model counting for {SAT}.
\newblock In \emph{Proc. of AAAI}, 2014.

\bibitem[Chakraborty et~al.(2015{\natexlab{a}})Chakraborty, Fremont, Meel,
  Seshia, and Vardi]{CFMSV15}
Supratik Chakraborty, Daniel~J. Fremont, Kuldeep~S. Meel, Sanjit~A. Seshia, and
  Moshe~Y. Vardi.
\newblock On parallel scalable uniform {SAT} witness generation.
\newblock In \emph{Proc. of TACAS}, 2015{\natexlab{a}}.

\bibitem[Chakraborty et~al.(2015{\natexlab{b}})Chakraborty, Fried, Meel, and
  Vardi]{CFMV15}
Supratik Chakraborty, Dror Fried, Kuldeep~S Meel, and Moshe~Y Vardi.
\newblock From weighted to unweighted model counting.
\newblock In \emph{Proc. of IJCAI}, 2015{\natexlab{b}}.

\bibitem[Chavira and Darwiche(2008)]{CD08}
Mark Chavira and Adnan Darwiche.
\newblock On probabilistic inference by weighted model counting.
\newblock \emph{Artificial Intelligence}, 2008.

\bibitem[Chen et~al.(2020)Chen, Jayaram, Levi, and Waingarten]{Jayaram20}
Xi~Chen, Rajesh Jayaram, Amit Levi, and Erik Waingarten.
\newblock Learning and testing junta distributions with subcube conditioning.
\newblock \emph{CoRR}, 2020.

\bibitem[Cohen-Steiner et~al.(2004)Cohen-Steiner, Alliez, and Desbrun]{CDAD04}
David Cohen-Steiner, Pierre Alliez, and Mathieu Desbrun.
\newblock Variational shape approximation.
\newblock In \emph{ACM SIGGRAPH Papers}. 2004.

\bibitem[Crigger and Khoury(2019)]{CK19}
Elliott Crigger and Christopher Khoury.
\newblock Making policy on augmented intelligence in health care.
\newblock \emph{AMA Journal of Ethics}, 2019.

\bibitem[Donohue(2019)]{D19}
Michael~E Donohue.
\newblock A replacement for {Justitia’s} scales?: Machine learning’s role
  in sentencing.
\newblock \emph{Harvard Journal of Law \& Technology}, 2019.

\bibitem[Dutra et~al.(2018)Dutra, Laeufer, Bachrach, and Sen]{DLBS18}
Rafael Dutra, Kevin Laeufer, Jonathan Bachrach, and Koushik Sen.
\newblock Efficient sampling of {SAT} solutions for testing.
\newblock In \emph{Proc. of ICSE}, 2018.

\bibitem[Ermon et~al.(2012)Ermon, Gomes, and Selman]{EGS12}
Stefano Ermon, Carla~P. Gomes, and Bart Selman.
\newblock Uniform solution sampling using a constraint solver as an oracle.
\newblock In \emph{Proc. of UAI}, 2012.

\bibitem[Ermon et~al.(2013)Ermon, Gomes, Sabharwal, and Selman]{EGSS13a}
Stefano Ermon, Carla~P Gomes, Ashish Sabharwal, and Bart Selman.
\newblock Embed and project: Discrete sampling with universal hashing.
\newblock In \emph{Proc. of NIPS}, 2013.

\bibitem[Gomes et~al.(2007)Gomes, Sabharwal, and Selman]{GSS07}
Carla~P. Gomes, Ashish Sabharwal, and Bart Selman.
\newblock Near-uniform sampling of combinatorial spaces using {XOR}
  constraints.
\newblock In \emph{Proc. of NIPS}, 2007.

\bibitem[Gupta et~al.(2019)Gupta, Sharma, Roy, and Meel]{GSRM19}
Rahul Gupta, Shubham Sharma, Subhajit Roy, and Kuldeep~S. Meel.
\newblock {WAPS: Weighted and Projected Sampling}.
\newblock In \emph{Proc. of TACAS}, 2019.

\bibitem[Jerrum(1998)]{J98}
Mark Jerrum.
\newblock Mathematical foundations of the markov chain monte carlo method.
\newblock In \emph{Probabilistic methods for algorithmic discrete mathematics}.
  1998.

\bibitem[Jerrum and Sinclair(1996)]{JS96}
Mark~R. Jerrum and Alistair Sinclair.
\newblock The {Markov} {Chain} {Monte Carlo} method: an approach to approximate
  counting and integration.
\newblock \emph{Approximation algorithms for NP-hard problems}, 1996.

\bibitem[Kamath and Tzamos(2019)]{KamathT19}
Gautam Kamath and Christos Tzamos.
\newblock Anaconda: {A} non-adaptive conditional sampling algorithm for
  distribution testing.
\newblock {SIAM}, 2019.

\bibitem[Kirkpatrick et~al.(1983)Kirkpatrick, Gelatt, and Vecchi]{KGV83}
Scott Kirkpatrick, C.~Daniel Gelatt, and Mario~P. Vecchi.
\newblock Optimization by simulated annealing.
\newblock \emph{Science}, 1983.

\bibitem[Maddison et~al.(2014)Maddison, Tarlow, and Minka]{MTM14}
Chris~J Maddison, Daniel Tarlow, and Tom Minka.
\newblock A* sampling.
\newblock In \emph{Proc. of NIPS}, 2014.

\bibitem[Madras(2002)]{M02}
Neal Madras.
\newblock Lectures on {Monte Carlo Methods}, {Fields Institute Monographs} 16.
\newblock \emph{American Mathematical Society}, 2002.

\bibitem[Meel et~al.(2016)Meel, Vardi, Chakraborty, Fremont, Seshia, Fried,
  Ivrii, and Malik]{MVC+16}
Kuldeep~S. Meel, Moshe~Y. Vardi, Supratik Chakraborty, Daniel~J Fremont,
  Sanjit~A Seshia, Dror Fried, Alexander Ivrii, and Sharad Malik.
\newblock Constrained sampling and counting: Universal hashing meets sat
  solving.
\newblock In \emph{AAAI Workshop: Beyond NP}, 2016.

\bibitem[Mosier and Skitka(2018)]{MS18}
Kathleen~L Mosier and Linda~J Skitka.
\newblock Human decision makers and automated decision aids: Made for each
  other?
\newblock In \emph{Automation and human performance}. 2018.

\bibitem[Murphy(2012)]{Mur12}
K.P. Murphy.
\newblock \emph{Machine Learning: A Probabilistic Perspective}.
\newblock MIT Press, 2012.

\bibitem[Seshia et~al.(2016)Seshia, Sadigh, and Sastry]{SSS16}
Sanjit~A Seshia, Dorsa Sadigh, and S~Shankar Sastry.
\newblock Towards verified artificial intelligence.
\newblock \emph{arXiv preprint arXiv:1606.08514}, 2016.

\bibitem[Valiant and Valiant(2011)]{ValiantV11}
Gregory Valiant and Paul Valiant.
\newblock {The Power of Linear Estimators}.
\newblock In \emph{Proc of {FOCS}}, 2011.

\bibitem[Valiant(2011)]{Valiant08}
Paul Valiant.
\newblock Testing symmetric properties of distributions.
\newblock \emph{{SIAM} J. Comput.}, 2011.

\bibitem[Wainwright and Jordan(2008)]{WJ2008}
Martin~J. Wainwright and Michael~I. Jordan.
\newblock Graphical models, exponential families, and variational inference.
\newblock \emph{Found. Trends Machine Learning}, 2008.

\end{thebibliography}
\clearpage
\appendix
\section*{Appendix}
%
%
%
%
%

\section{Relationship of {\barbarik} with Property Testing}\label{sec:propertytestingoverview}

Testing of samplers is basically testing if two distributions $\mathcal{D}_{\mathcal{G}(\varphi)}$ and $\mathcal{A}_{\mathcal{G}(\varphi)}$ are similar, where $\mathcal{G}$ is the sampler 
under test and $\mathcal{A}$ is the ideal sampler.  As stated in the Introduction and the Related Work section, the sub-field of property testing in theoretical computer science 
has been studying this problem for over two decades and our tester {\barbarik} draws ideas from some of the latest research in this area.  

In understanding the closeness between two distributions one may consider a variety of different distance measures. The variation distance (also called the $\ell_1$ distance) is possibly most commonly used. 
In property testing the problem is to distinguish between the case where the two distributions are $\varepsilon$-close in $\ell_1$ distance from the case where the distributions are $\eta$-far from each other in $\ell_1$ distance. 
An easier question, called the ``equivalence testing of distributions" considers the problem of distinguish identical distributions from distributions that are $\eta$-far from each other in $\ell_1$ distance. The former question, often 
referred to as the tolerant version of equivalence testing of distributions or estimation of variation distance, is more suitable for various applications. 
The goal in all the settings is to minimize the sample complexity. The time complexity or other complexity measures are usually not considered in property testing literature.  

The problem of equivalence testing of distributions was first considered by \cite{batu} and they (along with \cite{Valiant08} ) showed that the sample complexity was $\Theta(N^{2/3})$, where $N$ is the size of the support of the distributions. 
Note that, in the setting of samplers, $N$ is exponential in the input size and hence the number is prohibitively large. The tolerant version of the problem was proved to have even higher sample complexity of $\Theta(N)$ (\cite{Valiant08, ValiantV11}).
This was a significant bottleneck is practicality of these property testing algorithms and the tight lower bounds implied that no improvement was possible for algorithms that has only blackbox access to the distributions.  Even the much 
simpler problem of testing if a distribution is uniform requires $\Omega(\sqrt{N})$ samples. 

In \cite{conditional1, CanonneRS15} a new model for sampling was introduced called the conditional sampling. This model allowed access to the distributions that the standard sampling method (or the blackbox access to the distributions) could not give. 
It allowed a kind of grey-box access to the distributions. It was shown that in this model only $O(1/\varepsilon^2)$ conditional samples were needed to test if a distribution is uniform or $\varepsilon$-far from uniform. In fact similar conditional sample complexity 
is sufficient for the non-tolerant version of the equivalence testing of distributions. For  the tolerant version of  equivalence testing of distributions it was shown that polynomial in $\log(N)$ number of conditional samples suffice. Although this brings down the sample complexity drastically but still it was quite high for practical implementations.  On top of that  a major obstacle was whether the conditional samples were at all practical and were they implementable. 

In \cite{CM19} they successfully used the idea from the conditional sampling testing to test if samplers are uniform. They crucially used a special kind of conditional sampling. In \cite{CanonneRS15} a concept of pair-conditioning (they called PCOND) 
was introduced to define a restricted version of the conditional sampling model. A normal conditional sample is obtained by specifying a subset $S$ of the domain of the distribution $\mathcal{D}$ and then drawing a random sample from the conditional 
distribution $\mathcal{D}|_S$. A PCOND-sample is a normal conditional sample where the subset $S$ is of size $2$.  In  \cite{CM19} it was shown how this kind of restricted samples can be successfully implemented using a clever use of chain-formulas. 

When it come to the more general problems of non-tolerant version of  equivalence testing of distributions it can be shown that the  sample complexity in the PCOND-model is at least polynomial in $\log N$. The tolerant version has even higher 
PCOND-sample complexity. Since our primary objective was to have a tester that can be practical and implementable we had to circumvent the problem of high sample complexity and also of implementational issues of conditional sampling.
In our tester {\barbarik} we addressed these problems by using another trick from \cite{CM19}, that of, using two different notions of distance - $\ell_{\infty}$ for closeness and $\ell_1$ for farness.  In {\barbarik} we re-designed the 
sampler and give a proof of correctness in this paradigm using very different techniques as compare to that used in \cite{CM19}. 

It is worth noting here that recently conditional sampling and its various variants has been used to design efficient testing and learning algorithms for various other properties of distributions (\cite{clement, KamathT19,  BhattacharyyaC18, kane17, Jayaram20}). Many of these have the potential to be used more efficient and sophisticated testing of samplers and related questions. But the major question is the practicality of the models and the implementability of the algorithms.

\section{Proof of Correctness of {\barbarik}}\label{sec:proof}

In this section, we present the theoretical analysis of {\barbarik}, and the proof of Theorem~\ref{thm:main}. The proof clearly follows from the the following three lemmas.

\begin{restatable}[]{lemma}{complete}\label{thm:complete}
	If a sampler $\mathcal{G}$ is $\varepsilon$-close \footnote{for any $\varepsilon < \frac{1}{3}$ and $\eta > 6 \varepsilon$ } to the ideal sampler $\mathcal{A}$ then  {\barbarik} returns {\accept} with probability at least $(1-\delta)$.
\end{restatable}


\begin{restatable}[]{lemma}{sound}\label{thm:sound}
        If $\mathcal{G}$ is subquery consistent w.r.t {\fulcrumkernel} and if the distribution $\mathcal{D}_{\mathcal{G}(\varphi)}$ is
	$\eta$-far from the ideal sampler  then {\barbarik}
	returns {\reject} with  probability at least $(1-\delta)$.
\end{restatable}

\begin{restatable}[]{lemma}{query}\label{thm:query}
	Given $\varepsilon$, $\eta$ and $\delta$, {\barbarik} needs at most
	$\widetilde{O}\left(\frac{tilt(\mathtt{wt},\varphi)^2}{\eta(\eta - 6\varepsilon)^3}\right)$ samples for any input formula $\varphi$ and weight function $\mathtt{wt}$, where the tilde hides a poly logarithmic factor of $1/\delta, 1/\eta$ and
	$1/(\eta-6\varepsilon)$.
\end{restatable}

We will present the proofs of  Lemma~\ref{thm:complete}, Lemma~\ref{thm:sound} and Lemma~\ref{thm:query} in Section~\ref{sec:complete}, Section~\ref{sec:sound} and Section~\ref{sec:query} respectively.

In the rest of this section we will use the following notations: 
\begin{itemize}
\item We use $\mathbbm{1}(E)$ to represent the indicator variable for the event $E$. 
\item We use  $R_i$ to denote the event that {\barbarik} returns {\reject} in iteration $i$. 
\end{itemize}

For the proof of correctness of our algorithm, we need some standard concentration inequalities. The following versions of the Chernoff Bound will be used. 

\begin{lemma}\label{lem:relaxed_Chernoff}
	Let $Y_1, Y_2, \dots, Y_n$ be $i.i.d$  0-1 random variables.
	\begin{enumerate}
		\item If $\E[Y_i] \geq \theta \geq 0$, then for any $ t \leq \theta$,
		\begin{align*}
		\Pr	\left[  \sum_{j\in [n]} \frac{Y_j}{n}  \leq t  \right] < exp\left(-\frac{(\theta -t)^2n}{2\theta}\right)
		\end{align*}
		\item 
		If $\E[Y_i] \leq \theta $, then for any $t \geq \theta$,
		\begin{align*}
		\Pr	\left[  \sum_{j\in [n]} \frac{Y_j}{n}  \geq t  \right] <exp\left(-\frac{(t-\theta)^2n}{2t}\right)
		\end{align*}
	\end{enumerate}
\end{lemma}

We are now ready to present the proofs of Lemma~\ref{thm:complete}, Lemma~\ref{thm:sound} and Lemma~\ref{thm:query}.

\subsection{Proof of Lemma~\ref{thm:complete}}\label{sec:complete}
\complete*
 For the proof of Lemma~\ref{thm:complete} we will firstly show (in Lemma~\ref{lem:per_iteration_error}) that in each iteration of the loop, the probability that {\barbarik} returns {\reject} is less than $\delta/t$ and then the proof of  Lemma~\ref{thm:complete}  follows by the application of the Chernoff Bound. Recall that $R_i$ denotes the event that {\barbarik} returns {\reject} in iteration $i$.

\begin{restatable}[]{lemma}{periterationerror}\label{lem:per_iteration_error}
	If sampler $\mathcal{G}$ is $\varepsilon$-close to an ideal sampler $\mathcal{A}$ then the probability that {\barbarik} returns {\reject} in any particular iteration of the loop,  is atmost $\delta/t$. Then
	\begin{align}
	\Pr\left[ \overline{R_i}\mid \bigwedge_{j \in [i-1]} \overline{R_j}\right]  \geq \left(1 - \frac{\delta}{t}\right)
	\nonumber 
	\end{align}
\end{restatable}

\begin{proof}{(of Lemma~\ref{lem:per_iteration_error})}
	{\barbarik} returns {\reject} in the $i$th iteration if the $Bias$ (in the $i$th iteration) is more than $T$, where $T = \frac{L + H}{2}$ with $$L = \frac{(1+\varepsilon)p_{\mathcal{A}}(\varphi, S, \sigma_1) }{(1+\varepsilon)p_{\mathcal{A}}(\varphi, S, \sigma_1) + (1-\varepsilon)p_{\mathcal{A}}(\varphi, S, \sigma_2)}$$
	And since, by definition, all the elements in $\Gamma_1$, $\Gamma_2$ and $\Gamma_3$ are obtained by drawing independent samples from $\mathcal{D}_{\mathcal{G}(\varphi)}$, $\mathcal{D}_{\mathcal{A}(\varphi)}$ and  $\mathcal{D}_{\mathcal{G}(\hat{\varphi})}$ respectively so
	\begin{align}
	\Pr\left[\overline{R_i}\mid \bigwedge_{j \in [i-1]} \overline{R_j}\right] & = \Pr\left[\mbox{ $Bias \leq T$ in the $i$th iteration}\right] \nonumber \\
	& = 1 -  \Pr\left[\mbox{ $Bias > T$ in the $i$th iteration}\right] \nonumber \\
	& = 1 - \Pr\left[\sum_{j \in [N]}\frac{\mathbbm{1}(\Gamma_3[j]_{\downarrow S} = \sigma_1)}{N} > T \right] \nonumber
	\end{align}
	Note that the random variables $\mathbbm{1}(\Gamma_3[j]_{\downarrow S} = \sigma_1)$ are an i.i.d 0-1 random variable. And since the sampler $\mathcal{G}$ is assumed to be $\varepsilon$-close to the ideal sampler so we have
	$$ (1 - \varepsilon) p_{\mathcal{A}}(\hat{\varphi}, \Gamma_3[j]) \leq  p_{\mathcal{G}}(\hat{\varphi}, \Gamma_3[j]) \leq (1 + \varepsilon) p_{\mathcal{A}}(\hat{\varphi}, \Gamma_3[j]).$$
	Thus we have,
	\begin{align}
	\E[\mathbbm{1} (\Gamma_3[j]_{\downarrow S} = \sigma_1)] = p_{\mathcal{G}}(\hat{\varphi}, S, \sigma_1)  \leq (1 + \varepsilon)p_{\mathcal{A}}(\hat{\varphi}, S, \sigma_1)  \nonumber
	\end{align}
	Now, since $p_{\mathcal{A}}(\hat{\varphi}, S, \sigma_1) = \frac{p_{\mathcal{A}}(\varphi, S, \sigma_1)}{p_{\mathcal{A}}(\varphi, S, \sigma_1) + p_{\mathcal{A}}(\varphi, S, \sigma_2)}$ we have
	
	\begin{align}\label{eq:complete1}
	\E[\mathbbm{1} (\Gamma_3[j]_{\downarrow S} = \sigma_1)] = p_{\mathcal{G}}(\hat{\varphi}, S, \sigma_1) &\leq  \frac{(1+\varepsilon)p_{\mathcal{A}}(\varphi, S, \sigma_1)}{p_{\mathcal{A}}(\varphi, S, \sigma_1)+p_{\mathcal{A}}(\varphi, S, \sigma_2)}
	\end{align}
	Similarly, we have that
	\begin{align}\label{eq:complete2}
	\E[\mathbbm{1} (\Gamma_3[j]_{\downarrow S} = \sigma_2)] = p_{\mathcal{G}}(\hat{\varphi}, S, \sigma_2) &\geq \frac{(1-\varepsilon) p_{\mathcal{A}}(\varphi, S, \sigma_2)}{p_{\mathcal{A}}(\varphi, S, \sigma_1)+p_{\mathcal{A}}(\varphi, S, \sigma_2)}
	\end{align}
	Now we consider two cases depending on whether $p_{\mathcal{A}}(\varphi, S, \sigma_1) $ is greater or lesser than $ p_{\mathcal{A}}(\varphi, S, \sigma_2)$.
	If $p_{\mathcal{A}}(\varphi, S, \sigma_1) \leq  p_{\mathcal{A}}(\varphi, S, \sigma_2) $ then from Equation~\ref{eq:complete1} we have
	\begin{align}
	\E[\mathbbm{1} (\Gamma_3[j]_{\downarrow S} = \sigma_1)] & = p_{\mathcal{A}}(\hat{\varphi}, S, \sigma_1) \nonumber\\
	&\leq  \frac{(1+\varepsilon)p_{\mathcal{A}}(\varphi, S, \sigma_1)}{p_{\mathcal{A}}(\varphi, S, \sigma_1)+p_{\mathcal{A}}(\varphi, S, \sigma_2)} \nonumber\\
	& \leq \frac{(1+\varepsilon)p_{\mathcal{A}}(\varphi, S, \sigma_1)}{(1+\varepsilon)p_{\mathcal{A}}(\varphi, S, \sigma_1) + (1-\varepsilon)p_{\mathcal{A}}(\varphi, S, \sigma_2)} = L \label{eq:L_upperbound}
	\end{align}
	
	But if $p_{\mathcal{A}}(\varphi, S, \sigma_1) \geq p_{\mathcal{A}}(\varphi, S, \sigma_2) $ then from Equation~\ref{eq:complete1} we have
	
	\begin{align}
	\E[\mathbbm{1} (\Gamma_3[j]_{\downarrow S} = \sigma_2)] & = p_{\mathcal{A}}(\hat{\varphi}, S, \sigma_2)\nonumber\\
	&\geq \frac{(1-\varepsilon) p_{\mathcal{A}}(\varphi, S, \sigma_2)}{p_{\mathcal{A}}(\varphi, S, \sigma_1)+p_{\mathcal{A}}(\varphi, S, \sigma_2)}\nonumber\\
	& \geq \frac{(1-\varepsilon)p_{\mathcal{A}}(\varphi, S, \sigma_2)}{(1+\varepsilon)p_{\mathcal{A}}(\varphi, S, \sigma_1) + (1-\varepsilon)p_{\mathcal{A}}(\varphi, S, \sigma_2)}\nonumber
	\end{align}
	
	And in that case since $p_{\mathcal{A}}(\hat{\varphi}, S, \sigma_1) + p_{\mathcal{A}}(\hat{\varphi}, S, \sigma_2) = 1$ we have
	
	\begin{align}
	\E[\mathbbm{1} (\Gamma_3[j]_{\downarrow S} = \sigma_1)] & = p_{\mathcal{A}}(\hat{\varphi}, S, \sigma_1) \nonumber\\
	& = 1 - p_{\mathcal{A}}(\hat{\varphi}, S, \sigma_2) \nonumber\\
	& \leq 1 - \left(\frac{(1-\varepsilon)p_{\mathcal{A}}(\varphi, S, \sigma_2)}{(1+\varepsilon)p_{\mathcal{A}}(\varphi, S, \sigma_1) + (1-\varepsilon)p_{\mathcal{A}}(\varphi, S, \sigma_2)}\right) \nonumber\\
	& \leq  \frac{(1+\varepsilon)p_{\mathcal{A}}(\varphi, S, \sigma_1)}{(1+\varepsilon)p_{\mathcal{A}}(\varphi, S, \sigma_1) + (1-\varepsilon)p_{\mathcal{A}}(\varphi, S, \sigma_2)} = L\label{eq:L_lowerbound}
	\end{align}
	
	Thus in either case, from Equation (\ref{eq:L_upperbound}) and Equation (\ref{eq:L_lowerbound}) we have $\E[\mathbbm{1} (\Gamma_3[j]_{\downarrow S} = \sigma_1)] \leq L$.
	Now applying the Chernoff bound from Lemma \ref{lem:relaxed_Chernoff} we have
	\begin{align}
	\Pr\left[Bias \geq T \right] & = \Pr\left[\sum_{j \in [N]}\frac{\mathbbm{1}(\Gamma_3[j]_{\downarrow S} = \sigma_1)}{N} > T \right]\nonumber\\
	& = exp\left(-\frac{(T-L)^2N}{2L}\right)  = exp\left(-\frac{(H-L)^2N}{8L}\right) \nonumber\\
	& \leq exp\left(-\frac{(H-L)^2N}{8H}\right) \quad \mbox{ Because [$H \geq L$}] \label{line:HL}\\
	& \leq \frac{\delta}{t} \label{eq:delta_by_t},
	\end{align}
	where the inequality in line (\ref{line:HL}) follows because  $H \geq L$ when $\varepsilon\leq 1/3$ and $\eta\geq 6\varepsilon$\footnote{$H\geq L$ if $hi\geq lo$ that is $(6\varepsilon +\eta)/4 \geq (2\varepsilon)/(1-\varepsilon)$}
	and last inequality follows because $N = n.H/(H-L)^2$ where $n = 8\log (t/\delta)$.
\end{proof}

\begin{proof}{(of Lemma~\ref{thm:complete})}
	Let $R_i$ denote the event that {\barbarik} returns {\reject} in iteration $i$
	and $\overline{R}$ denote the event that {\barbarik} returns {\accept}. Thus $\overline{R} = \cap_{i}\overline{R_i}$.
	
	In the $i^{th}$ iteration if the bias is less than the threshold, {\barbarik} fails to {\reject}. Thus from Lemma~\ref{lem:per_iteration_error} if the sampler $\mathcal{G}$ is $\varepsilon$-close to the ideal sampler $\mathcal{A}$ then
	\begin{align}
	\Pr\left[ \overline{R_i} {\mid} \bigwedge_{j \in [i-1]} \overline{R_j}\right]  \geq 1-\frac{\delta}{t}\nonumber	
	\end{align}
	If {\barbarik} has not returned {\reject} in any of the iteration then after the last iteration {\barbarik} returns {\accept}. The probability of {\barbarik} returning {\accept} (event $\overline{R}$) is
	\begin{align}
	\Pr\left[\overline{R}\right] & \geq  \prod_{i\in[t]}\Pr\left[\overline{R_i} {\mid} \bigwedge_{j \in [i-1]}\overline{R_j}\right]\nonumber
	\geq \left(1- \frac{\delta}{t}\right)^t \geq 1-\delta \nonumber
	\end{align}
\end{proof}

\subsection{Proof of Lemma~\ref{thm:sound}}\label{sec:sound}
\sound*
\begin{proof}
	To prove the Lemma, we will start by splitting the set $R_{\varphi}$ into disjoint subsets depending on the distribution $D_{\mathcal{G}(\varphi)}$.
	
	\begin{definition}
		We define the following sets for use in the soundness proof:
		\begin{itemize}
			\item $D = \{x \in R_{\varphi} :\ p_{\mathcal{G}}(\varphi, x) \leq p_{\mathcal{A}}(\varphi, x)\}$
			\item $U = R_\varphi\setminus D$
			\item $U_0 = \{x \in R_{\varphi} :\ p_{\mathcal{A}}(\varphi, x) < p_{\mathcal{G}}(\varphi, x) \leq \left( 1 + \frac{\eta + 6\varepsilon}{4}\right)p_{\mathcal{A}}(\varphi, x) \} $. 	
			\item $ U_1 = \{x \in R_{\varphi} :\   \left( 1 + \frac{\eta + 6\varepsilon}{4}\right)p_{\mathcal{A}}(\varphi, x) < p_{\mathcal{G}}(\varphi, x)  \} $
		\end{itemize}
	\end{definition}

	Recall, $R_i$ is the event that {\barbarik} returns {\reject} in the $i$th iteration of the for loop. Then the following lemmas helps us to lower bound the probability of $\Gamma_1[i] \in U_1 \wedge \Gamma_2[i] \in D$ and the probability of $R_i$ under the condition that $\Gamma_1[i] \in U_1 \wedge \Gamma_2[i] \in D$.

	\begin{restatable}[]{lemma}{probofaccept}\label{lem:probofaccept}
		If the sampler $\mathcal{G}$ is $\eta$-far from the ideal sampler then
	$$\Pr\left[ R_i {\mid}  (\bigwedge_{j \in [i-1]}\overline{R_j})\wedge (\Gamma_1[i] \in U_1 \wedge \Gamma_2[i] \in D)\right] \geq \frac{4}{5}.$$
\end{restatable}

\begin{restatable}[]{lemma}{rounds}\label{lem:rounds}
	If the sampler $\mathcal{G}$ is $\eta$-far from the ideal sampler on input $\varphi$ then $$\Pr \left[\Gamma_1[i] \in U_1 \wedge \Gamma_2[i] \in D \right]  \geq \frac{\eta(\eta-6\varepsilon)}{8}.$$
\end{restatable}

	And now using Lemmas~\ref{lem:rounds} and \ref{lem:probofaccept} we can complete the proof of soundness.
	The probability that {\barbarik} returns {\reject} in the $i$th iteration of the for loop is
	
	\begin{align}
	&\Pr\left[R_i \mid \bigwedge_{j \in [i-1]}\overline{R_j} \right]  \nonumber \\
	= & \Pr\left[ R_i \mid  (\bigwedge_{j \in [i-1]}\overline{R_j})\wedge (\Gamma_1[i] \in U_1 \wedge \Gamma_2[i] \in D)\right]
	\cdot \Pr[\Gamma_1[i] \in U_1 \wedge \Gamma_2[i] \in D]\nonumber\\
	\geq &   \left( \frac{4}{5}\right)\frac{\eta(\eta-6\varepsilon)}{8}\label{line:false_reject_in_one_round} \quad [\text{From Lemma~\ref{lem:rounds} and Lemma~\ref{lem:probofaccept}}]
	\end{align}
	The probability of {\barbarik} returning {\reject}  in any iteration (event ${R}$) is given by
	\begin{align*}
	\Pr\left[\cup_i R_i\right]&= 1- \prod_{i\in[t]}\Pr\left[\overline{R_i} \mid \bigwedge_{j \in [i-1]}\overline{R_j}\right]\\
	&\geq 1-  \prod_{i \in [t]}\left(1-\frac{\eta(\eta-6\varepsilon)}{10}\right) \quad [\text{Using Equation (\ref{line:false_reject_in_one_round})}]\\
	&\geq 1- \left(1-\frac{\eta(\eta-6\varepsilon)}{10}\right)^t\\
	\text{Substituting $t$,\quad}&\geq 1-\delta
	\end{align*}
\end{proof}


Now to complete the proof of Lemma~\ref{thm:sound}  we have to prove the Lemma~\ref{lem:probofaccept} and  Lemma~\ref{lem:rounds}. They are presented next. 

	

\probofaccept*

	\begin{proof}{(of Lemma~\ref{lem:probofaccept})}
	Let us assume $\Gamma_1[i] \in U_1$ and  $\Gamma_2[i]\in D$. That is, we have $p_\mathcal{G}(\varphi, S, \Gamma_2[i]) \leq p_\mathcal{A}(\varphi, S, \Gamma_2[i])$ and
	$p_\mathcal{G}(\varphi, S, \Gamma_1[i] ) > \left(1 + \frac{\eta + 6\varepsilon}{4}\right)p_\mathcal{A}(\varphi, S, \Gamma_1[i] )$. It follows that
	\begin{align}
	\frac{p_\mathcal{G}(\varphi, S, \Gamma_1[i])}{p_\mathcal{G}(\varphi, S, \Gamma_2[i])} & \geq \left(1 + \frac{6\varepsilon+\eta}{4}\right)\cdot \frac{p_\mathcal{A}({\varphi}, S, \Gamma_1[i] )}{p_\mathcal{A}({\varphi}, S, \Gamma_2[i])}\label{eq:sound2ndLemma}
	\end{align}
	Since $\forall x>0, a/b>x \implies  a/(a+b)>x/(x+1) $, we have from Equation~\ref{eq:sound2ndLemma}
	\begin{align*}
	& \frac{p_\mathcal{G}(\varphi, S, \Gamma_1[i] )}{p_\mathcal{G}(\varphi, S, \Gamma_2[i]) + p_\mathcal{G}(\varphi, S, \Gamma_1[i])} \\
	\geq & \left(1 + \frac{6\varepsilon+\eta}{4}\right)\cdot \frac{p_\mathcal{A}({\varphi}, S, \Gamma_1[i] )}{p_\mathcal{A}({\varphi}, S, \Gamma_2[i])} \cdot \left(1+\left(1 + \frac{6\varepsilon+\eta}{4}\right)\cdot \frac{p_\mathcal{A}({\varphi}, S, \Gamma_1[i] )}{p_\mathcal{A}({\varphi}, S, \Gamma_2[i])}\right)^{-1}
	\end{align*}
	Thus we have
	\begin{align}
	& \E[\mathbbm{1}(\Gamma_3[j]_{\downarrow S} = \sigma_1)]  =  p_{\mathcal{G}}(\hat{\varphi}, S, \Gamma_1[i]) \nonumber \\
	= & \frac{p_\mathcal{G}(\varphi, S, \Gamma_1[i] )}{p_\mathcal{G}(\varphi, S, \Gamma_2[i]) + p_\mathcal{G}(\varphi, S, \Gamma_1[i])}\nonumber \quad [\mbox{ by the subquery consistent sampler assumption}]\\
	\geq & \left(1 + \frac{6\varepsilon+\eta}{4}\right)\cdot \frac{p_\mathcal{A}({\varphi}, S, \Gamma_1[i] )}{p_\mathcal{A}({\varphi}, S, \Gamma_2[i])} \cdot \left(1+\left(1 + \frac{6\varepsilon+\eta}{4}\right)\cdot \frac{p_\mathcal{A}({\varphi}, S, \Gamma_1[i] )}{p_\mathcal{A}({\varphi}, S, \Gamma_2[i])}\right)^{-1}\nonumber \\
	= & H \quad [\mbox{By definition of $H$}]  \label{eq:H}
	\end{align}	
	{\barbarik} returns {\reject} in the $i$th iteration if the $Bias$ (in the $i$th iteration) is more than $T$, where $T = \frac{L + H}{2}$ with $$H = \frac{(1+\frac{6\varepsilon+\eta}{4})p_{\mathcal{A}}(\varphi, S, \sigma_1) }{(1+\frac{6\varepsilon+\eta}{4})p_{\mathcal{A}}(\varphi, S, \sigma_1) + p_{\mathcal{A}}(\varphi, S, \sigma_2)}$$
	And since, by definition, all the elements in $\Gamma_1$, $\Gamma_2$ and $\Gamma_3$ are obtained by drawing independent samples from $\mathcal{D}_{\mathcal{G}(\varphi)}$, $\mathcal{D}_{\mathcal{A}(\varphi)}$ and  $\mathcal{D}_{\mathcal{G}(\hat{\varphi})}$ respectively so
	\begin{align*}
	&\Pr\left[R_i\mid (\bigwedge_{j \in [i-1]} \overline{R_j})\bigwedge (\Gamma_1[i]\in U_1 \wedge \Gamma_2[i]\in D)\right] \\
	= & \Pr\left[\mbox{ $Bias > T$ in the $i$th iteration}\mid (\Gamma_1[i]\in U_1 \wedge \Gamma_2[i]\in D) \right] \\
	= & \Pr\left[\sum_{j \in [N]}\frac{\mathbbm{1}(\Gamma_3[j]_{\downarrow S} = \sigma_1)}{N} \geq T \mid (\Gamma_1[i]\in U_1 \wedge \Gamma_2[i]\in D) \right]
	\end{align*}
	Now since  $\mathbbm{1}(\Gamma_3[j]_{\downarrow S} = \sigma_1)$ are i.i.d 0-1 random variables and since $\Gamma_1[i]\in U_1$ and $\Gamma_2[i]\in D$ implies $\E[\mathbbm{1}(\Gamma_3[j]_{\downarrow S} = \sigma_1)] \geq H$ (from Equation~\ref{eq:H}) by applying Chernoff bound from Lemma \ref{lem:relaxed_Chernoff}  we have: %
	\begin{align*}
	\Pr\left[\frac{1}{N}\sum_{j \in [N]}\mathbbm{1}(\Gamma_3[j]_{\downarrow S} = \sigma_1) \geq T \right]  &\leq exp\left(-\frac{(H-T)^2N}{8H}\right) \\
	\text{by the choice of N \quad }&\leq \frac{\delta}{t}\\
	\text{since $\delta < 0.5$ and $t\geq 3$}  \quad  &\leq 1/5
	\end{align*}
\end{proof}



\rounds*

	\begin{proof}{of Lemma~\ref{lem:rounds})}
	Since the sampler $\mathcal{G}$ is $\varepsilon$-far from the ideal sampler on input $\varphi$ so, the $\ell_1$ distance between $\mathcal{D}_{\mathcal{G}}(\varphi)$ and  $\mathcal{D}_{\mathcal{A}}(\varphi)$ is at least $\eta$. By the definition of sets $U$ and $D$ we have,
	\begin{align}
	\sum_{x\in {U }} (p_{\mathcal{G}}(\varphi, x) - p_{\mathcal{A}}(\varphi, x)) = \sum_{x\in {D}}(p_{\mathcal{A}}(\varphi, x) - p_{\mathcal{G}}(\varphi, x))& \geq \frac{\eta}{2} \label{line:high1_weight}
	\end{align}
	
	Now by definition of $U_0$, we have
	\begin{align}
	\sum_{x\in U_0} (p_{\mathcal{G}}(\varphi, x) -p_{\mathcal{A}}(\varphi, x)) < \frac{\eta + 6\varepsilon}{4}\sum_{x\in U_0} p_{\mathcal{A}}(\varphi, x) < \frac{\eta+6\varepsilon}{4} \label{line:high0_weight}
	\end{align} As $U= U_0 \cup U_1$,
	\begin{align}
	&\sum_{x\in U_1}\left(p_{\mathcal{G}}(\varphi, x) - p_{\mathcal{A}}(\varphi, x)\right) \nonumber \\
	= &\sum_{x\in U}\left(p_{\mathcal{G}}(\varphi, x) - p_{\mathcal{A}}(\varphi, x)\right)-\sum_{x\in U_0}\left(p_{\mathcal{G}}(\varphi, x) - p_{\mathcal{A}}(\varphi, x)\right)   \label{line:union_of_high} \end{align}
	Substituting Equation (\ref{line:high0_weight}) and Equation (\ref{line:high1_weight}) in Equation (\ref{line:union_of_high}) we get:-
	\begin{align*}
	\sum_{x\in U_1}\left(p_{\mathcal{G}}(\varphi, x) - p_{\mathcal{A}}(\varphi, x)\right) &\geq \frac{\eta}{2} - \frac{\eta+6\varepsilon}{4} =\frac{\eta-6\varepsilon}{4} \\\text{Therefore,}
	\sum_{x\in U_1}p_{\mathcal{G}}(\varphi, x) &\geq \frac{\eta-6\varepsilon}{4}
	\end{align*}	
	Thus we have,
	\begin{align}
	\Pr\left[ \Gamma_1[i] \in U_1\right] = \sum_{x\in U_{1}}p_{\mathcal{G}}(\varphi, x)  &\geq \frac{\eta-6\varepsilon}{4}
	\label{line:prob_high}
	\end{align}
	From Equation (\ref{line:high1_weight}) we know that,
	\begin{align}
	\Pr\left[\Gamma_2[i] \in D\right] = \sum_{x\in D} p_{\mathcal{A}}(\varphi, x)  &\geq \frac{\eta}{2}
	\label{line:prob_low}
	\end{align}
	Since $\Gamma_1[i] \in U_1$ and $\Gamma_2[i] \in D$ are independent events, putting together Equation  (\ref{line:prob_high}) and Equation (\ref{line:prob_low}), we see that
	\begin{align*}
	\Pr\left[ \Gamma_1[i] \in U_1 \wedge \Gamma_2[i] \in D  \right]  \geq \frac{\eta(\eta-6\varepsilon)}{8}
	\end{align*}
\end{proof}

\subsection{Proof of Lemma~\ref{thm:query}}\label{sec:query}
\query*
	\begin{proof}
	From Algorithm \ref{algo:theory_new}, line \ref{line:trials}, we see that the number of trials is:
	\begin{align*}
	t &= \frac{ln(1/\delta)}{ln\left(\frac{10}{10-\eta(\eta-6\varepsilon)}\right)}\\
	\text{($ln(x) \leq x-1$)}\quad  t&\leq ln(1/\delta) \frac{10}{(\eta(\eta- 6\varepsilon))}
	\end{align*}
	In every iteration we calculate a value $N$ according to the expression:
	\begin{align}
	N &= 8ln\left(\frac{t}{\delta}\right)\cdot \frac{\alpha \cdot hi}{1+ \alpha\cdot  hi} \cdot \left(\frac{\alpha \cdot hi}{1+\alpha \cdot hi}-\frac{\alpha \cdot lo}{1+\alpha \cdot lo}\right)^{-2}\nonumber \\
	&= 8ln\left(\frac{t}{\delta}\right)\cdot  \left(\frac{1}{hi-lo}\right)^2\cdot  hi \cdot \frac{1+\alpha \cdot hi}{\alpha}\cdot(1+\alpha \cdot lo)^2\nonumber\\
	\text{ ($1 < lo < hi < 2$)\quad}&<8ln\left(\frac{t}{\delta}\right)\cdot  \left(\frac{1}{hi-lo}\right)^2\cdot  2\cdot \frac{1+\alpha \cdot 2}{\alpha}\cdot(1+\alpha \cdot 2)^2\nonumber	
	\end{align}
	On Line (\ref{line:alpha}) in Algorithm \ref{algo:theory_new} we define:
	\begin{align*}
	\alpha &= \frac{\mathtt{wt}(\sigma_1)}{\mathtt{wt}( \sigma_2)}\\
	\text{(Definition \ref{defn:tilt})}\quad tilt(\mathtt{wt},\varphi)&=
	\underset{\sigma_1,\sigma_2 \in R_\varphi}{\max}\; \frac{\mathtt{wt}(\sigma_1)}{\mathtt{wt}(\sigma_2)}
	\end{align*}
	Thus, $\alpha \leq tilt(\mathtt{wt},\varphi)$. Substituting the values of $\alpha, lo$ and $hi$, we get:
	\begin{align}
	N < 8ln\left(\frac{t}{\delta}\right)\cdot \left(\frac{tilt(\mathtt{wt},\varphi)}{\eta-6\varepsilon}\right)^2 \nonumber
	\end{align}  \\
	The maximum number of samples drawn after $t$ trials is:
	\begin{align}
	2t+tN &< 2tN  \nonumber \\
	\text{(Substituting for t,$N$)\quad \quad} &< 8ln\left(\frac{1}{\delta}\cdot \frac{10\cdot ln(1/\delta)}{\eta(\eta - 6\varepsilon)} \right)\times  \frac{10\cdot ln(1/\delta)}{\eta(\eta - 6\varepsilon)}\times\frac{tilt(\mathtt{wt},\varphi)^2}{(\eta-6\varepsilon)^2} \nonumber\\
	&=\Tilde O\left(\frac{tilt(\mathtt{wt},\varphi)^2}{\eta(\eta-6\varepsilon)^3}\right)\nonumber
	\end{align}
\end{proof}

\section{Log-Linear Distributions and Inverse Transform Sampling}\label{sec:weighted_dist}

Log-linear models capture wide class of distributions of interest including those arising from graphical models, conditional random fields, skip-gram models~\cite{Mur12}. Formally, for $\sigma \in \{0,1\}^n$, we define 
\begin{align*}
 \Pr[\sigma | \theta] \propto e^{\theta \cdot \sigma}  
\end{align*}
Following Chavira and Darwiche~\cite{CD08},  we describe the following equivalent representation, called literal-weighted functions, of log-linear models. 
\begin{definition}[Literal-Weighted Functions]\label{def:literal_weighted}
	For a CNF formula $\varphi$ and set $S\subseteq Supp(\varphi)$, a weight function $\mathtt{wt}: \{0,1\}^{|S|}\to (0,1)$ is called a literal-weighted function if there is a map $\mathtt{W}: S \to (0,1)$ such that for any assignment $\sigma \in R_{\varphi_{\downarrow{S}}}$
	\begin{align*}\wt{\sigma} = \prod_{x \in \sigma}
	\begin{cases}
	\mathtt{W}(x) &if \quad x = 1\\
	1-\mathtt{W}(x) &if \quad x =0
	\end{cases}
	\end{align*}
	
	In this case we call $\mathtt{wt}$ a literal-weighted function w.r.t. $\mathtt{W}$. And note that we have $\Pr[ \sigma ] \propto \wt{\sigma}$.  
\end{definition}

We now discuss the standard technique of inverse transform sampling for completeness. For completeness, we follow the description due to Chakrborty et al~\cite{CFMV15}. 
\begin{restatable}[]{lemma}{equivalence}\label{thm:replacing_samplers}
	For any $\varepsilon$-close uniform sampler $\mathcal{V}$, any CNF formula $\varphi$ with support $S$  and a literal-weighted function $\mathsf{wt}: \{0,1\}^{|S|} \rightarrow (0,1)$, we can construct a $\hat\varphi$ s.t.
	\begin{align*}
	\forall_{\sigma \in R_{\varphi}},
	\frac{(1-\varepsilon)\mathtt{wt}(\sigma)}{\sum_{\sigma' \in R_{\varphi}}\mathtt{wt}(\sigma')} \leq p_\mathcal{V}(\hat\varphi,S,\sigma) \leq \frac{(1+\varepsilon)\mathtt{wt}(\sigma)}{\sum_{\sigma' \in R_{\varphi}}\mathtt{wt}(\sigma')}
	\end{align*}
\end{restatable}
\begin{proof}
	Let $S_i= \{x_{i,1},\cdots, x_{i,m_i}\}$ be a set of $m_i$ ``fresh” variables (i.e. variables that were not used before) for each $x_i\in S$. 	Given any integer $m_i>0$ and a positive odd number $k_i< 2^{m_i}$, we construct  $\varphi_{k_i, m_i}(x_{i,1}, \cdots x_{i,m_i})$ using the chain formula construction in~\cite{CFMV15} such that $|R_{\varphi_{k_i,m_i}}| = k$. For notational clarity, we simply write $\varphi_{k_i, m_i}$ when the arguments of the chain formula are clear from context. For each variable $x_i \in S$, such that $\mathtt{W}(x_i^1)=\frac{k_i}{2^{m_i}}$, and $\mathtt{W}(x_i^0)=1-\mathtt{W}(x_i)$, let $(x_i \leftrightarrow\varphi_{k_i, m_i})$ be the representative clause. Thus let $\varphi^{CNF}=\bigwedge_{i \in S}(x_i \leftrightarrow\varphi_{k_i,m_i})$. We then define the formula $\hat\varphi$ as follows:
	\begin{align*}
	\hat\varphi = \varphi \wedge \varphi^{CNF}
	\end{align*}
	We can see that model count of the formula $|R_{\hat\varphi}|$ can be given by:
	\begin{align}
	|R_{\hat\varphi}| &= \sum_{\hat\sigma\in R_{\hat\varphi}}1=\sum_{\sigma\in R_{\varphi}} \sum_{(\hat\sigma \in R_{\hat\varphi}:\hat\sigma_{\downarrow{S}}=\sigma)}1
	\label{line:model_count_of_chain_formula}
	\end{align}
	Since the representative formula of every variable uses a fresh set of variables, we have from the structure of $\hat\varphi$ that if $\sigma$ is a witness of $\varphi$ then:
	\begin{align}\label{line:product_expansion}
	\sum_{(\hat\sigma \in R_{\hat\varphi}:\hat\sigma_{\downarrow{S}}=\sigma)}1 = \prod_{i \in \sigma^0 }(2^{m_i}-k_i) \prod_{i \in \sigma^1}k_i
	\end{align}
	For any $\sigma \in R_{\varphi}$:
	\begin{align} \nonumber
	p_\mathcal{U}(\hat\varphi, S, \sigma) &= \sum_{(\hat\sigma \in R_{\hat\varphi}:\hat\sigma_{\downarrow{S}}=\sigma)}p_\mathcal{U}(\hat\varphi, \hat{S}, \hat\sigma) \\ \nonumber
	&= \sum_{(\hat\sigma \in R_{\hat\varphi}:\hat\sigma_{\downarrow{S}}=\sigma)}\frac{1}{|R_{\hat\varphi}|}\\ \nonumber
	&=\frac{\sum_{(\hat\sigma \in R_{\hat\varphi}:\hat\sigma_{\downarrow{S}} = \sigma)}1}{\sum_{\sigma'\in R_{\varphi}} \sum_{(\hat\sigma \in R_{\hat\varphi}:\hat\sigma_{\downarrow{S}}= \sigma')}1} \quad  \text{Using (\ref{line:model_count_of_chain_formula})}\\ \nonumber
	&=\frac{\prod_{i \in \sigma^0 }(2^{m_i}-k_i) \prod_{i \in \sigma^1}k_i}{\sum_{\sigma' \in R_{\varphi}}\prod_{i \in \sigma'^0 }(2^{m_i}-k_i) \prod_{i \in \sigma'^1}k_i} \quad \text{Using (\ref{line:product_expansion})}\\ \nonumber
	&=\frac{\prod_{i \in \sigma^0 }(2^{m_i}-k_i) \prod_{i \in \sigma^1}k_i}{\prod_{i\in S} 2^{m_i}} \cdot \frac{\prod_{i\in S} 2^{m_i}}{\sum_{\sigma'\in R_{\varphi}}\prod_{i \in \sigma'^0 }(2^{m_i}-k_i) \prod_{i \in \sigma'^1}k_i}\\ \nonumber
	&=\frac{\prod_{i\in S}\mathtt{W}(\sigma_{\downarrow{x_i}})}{\sum_{\sigma'\in R_{\varphi}}\prod_{i \in S}\mathtt{W}(\sigma'_{\downarrow{x_i}})}\\
	&=\frac{\mathtt{wt}(\sigma)}{\sum_{\sigma' \in R_{\varphi}}\mathtt{wt}(\sigma')}\label{line:equality}
	\end{align}	
	From the definition of $\varepsilon$-additive closeness (Def. \ref{def:closeness_and_farness}) we have:
	\begin{align*}
	(1-\varepsilon)p_\mathcal{U}(\varphi,S, \sigma )\leq p_\mathcal{V}(\varphi,S,\sigma) \leq (1+\varepsilon)p_\mathcal{U}(\varphi,S,\sigma)
	\end{align*}
	Substituing into \ref{line:equality}, we get:
	\begin{align*}
	\forall_{\sigma \in R_{\varphi}},
	\frac{(1-\varepsilon)\mathtt{wt}(\sigma)}{\sum_{\sigma' \in R_{\varphi}}\mathtt{wt}(\sigma')} \leq p_\mathcal{V}(\hat\varphi,S,\sigma) \leq \frac{(1+\varepsilon)\mathtt{wt}(\sigma)}{\sum_{\sigma' \in R_{\varphi}}\mathtt{wt}(\sigma')}
	\end{align*}
\end{proof}

\begin{remark}
  It is worth noting that Lemma~\ref{thm:replacing_samplers} implies that if $\mathcal{V}$ is $\varepsilon$-close uniform sampler $\mathcal{V}$ then it can be used as a blackbox to obtain a
  $\varepsilon$-close to an ideal sampler w.r.t any literal-weighted function $\mathsf{wt}$. 
  It should also be noted that  Lemma~\ref{thm:replacing_samplers} does not imply that if $\mathcal{V}$ is $\eta$-far from a uniform sampler, then the new sampler (obtained using the above
  transformation) is also far from the ideal sampler w.r.t $\mathsf{wt}$. 
  Therefore, to test whether $p_\mathcal{V}(\hat\varphi,S,\sigma)$ is close to ideal sampler, one can not rely on merely testing uniformity of $\mathcal{V}$. 
\end{remark}

\section{Extended Tables of Results}\label{sec:detailedresults}

\subsection{Comparing sample complexity.}
``A''(``R'') represent {\barbarik} returning {\accept}({\reject}). ``DNS'' is used against those instances on which the indicated sampler Did Not Sample. ``-'' indicates that {\barbarik} timed out on that particular instance on the indicated sampler. Note that ``DNS'' is different from ``-'' as ``DNS'' indicates the failure of the underlying sampler to sample the initial set of samples, while ``-'' indicates the failure of {\barbarik} to finish within the timeout period. The timeout was set to 50,000 seconds for {\wSearchTreeSampler} and {\wQuickSampler}, while for {\wWeightGen} it was 24 hours.
\begin{longtable}[ht]
{l>{\centering}m{0.15\linewidth}>{\centering\arraybackslash}m{0.16\linewidth}>{\centering\arraybackslash}m{0.16\linewidth}>{\centering\arraybackslash}m{0.16\linewidth}}
\caption{The Extended Table}\\ 
  \toprule  && \multicolumn{3}{c}{\barbarik}  \\ \cmidrule(l){3-5} 
 Benchmark&\shortstack{$tilt$\\ (maxSamp)}&\shortstack{\wWeightGen\\ (samples)}&\shortstack{\wSearchTreeSampler\\ (samples)}&\shortstack{\wQuickSampler\\ (samples)}\\
	\midrule 
 \endfirsthead 
	\hline
\multicolumn{5}{r@{}}{continued \ldots}\\ 
\endfoot
\hline
\endlastfoot
 \toprule
  && \multicolumn{3}{c}{\barbarik} \\ \cmidrule(l){3-5}	Benchmark&\shortstack{$tilt$\\ (maxSamp)}&\shortstack{\wWeightGen\\ (samples)}&\shortstack{\wSearchTreeSampler\\ (samples)}&\shortstack{\wQuickSampler\\ (samples)}\\ 
\midrule\endhead107.sk\_3\_90&\shortstack{1\\ (2e+05)}&DNS&\shortstack{R\\ (5146)}&\shortstack{R\\ (6009)}\\ \midrule 
tableBasedAddition.sk&\shortstack{1\\ (2e+05)}&DNS&\shortstack{R\\ (6009)}&\shortstack{R\\ (24534)}\\ \midrule 
55.sk\_3\_46&\shortstack{1\\ (2e+05)}&DNS&\shortstack{R\\ (8911)}&\shortstack{R\\ (4354)}\\ \midrule 
111.sk\_2\_36&\shortstack{1\\ (2e+05)}&DNS&\shortstack{R\\ (23543)}&\shortstack{R\\ (5150)}\\ \midrule 
17.sk\_3\_45&\shortstack{1\\ (2e+05)}&DNS&\shortstack{R\\ (1e+05)}&\shortstack{R\\ (4677)}\\ \midrule 
80.sk\_2\_48&\shortstack{1\\ (2e+05)}&DNS&\shortstack{R\\ (4284)}&\shortstack{R\\ (4627)}\\ \midrule 
27.sk\_3\_32&\shortstack{1\\ (2e+05)}&\shortstack{A\\ (1e+05)}&\shortstack{R\\ (25329)}&\shortstack{R\\ (6009)}\\ \midrule 
70.sk\_3\_40&\shortstack{1\\ (2e+05)}&DNS&\shortstack{R\\ (10402)}&\shortstack{R\\ (17704)}\\ \midrule 
32.sk\_4\_38&\shortstack{1\\ (2e+05)}&\shortstack{A\\ (1e+05)}&\shortstack{R\\ (18081)}&\shortstack{R\\ (14682)}\\ \midrule 
84.sk\_4\_77&\shortstack{1\\ (2e+05)}&DNS&\shortstack{R\\ (5146)}&\shortstack{R\\ (4354)}\\ \midrule 
53.sk\_4\_32&\shortstack{1\\ (2e+05)}&\shortstack{A\\ (1e+05)}&\shortstack{R\\ (35618)}&\shortstack{R\\ (6009)}\\ \midrule 
s35932\_3\_2&\shortstack{3\\ (6e+05)}&DNS&TO&\shortstack{R\\ (11756)}\\ \midrule 
s35932\_7\_4&\shortstack{3\\ (6e+05)}&DNS&TO&\shortstack{R\\ (11756)}\\ \midrule 
s832a\_3\_2&\shortstack{3\\ (6e+05)}&\shortstack{A\\ (1e+05)}&\shortstack{R\\ (8708)}&\shortstack{R\\ (54138)}\\ \midrule 
109.sk\_4\_36&\shortstack{8\\ (3e+06)}&DNS&\shortstack{R\\ (26218)}&\shortstack{R\\ (6009)}\\ \midrule 
77.sk\_3\_44&\shortstack{11\\ (5e+06)}&DNS&\shortstack{R\\ (47582)}&\shortstack{R\\ (47907)}\\ \midrule 
s35932\_15\_7&\shortstack{12\\ (6e+06)}&DNS&TO&\shortstack{R\\ (4354)}\\ \midrule 
s832a\_7\_4&\shortstack{15\\ (8e+06)}&\shortstack{A\\ (1e+05)}&\shortstack{R\\ (4393)}&\shortstack{R\\ (13350)}\\ \midrule 
51.sk\_4\_38&\shortstack{18\\ (1e+07)}&\shortstack{A\\ (78661)}&\shortstack{R\\ (4284)}&\shortstack{R\\ (4627)}\\ \midrule 
29.sk\_3\_45&\shortstack{26\\ (2e+07)}&DNS&\shortstack{R\\ (4284)}&\shortstack{R\\ (55989)}\\ \midrule 
81.sk\_5\_51&\shortstack{27\\ (3e+07)}&DNS&\shortstack{R\\ (28409)}&\shortstack{A\\ (2e+05)}\\ \midrule 
s349\_3\_2&\shortstack{28\\ (3e+07)}&\shortstack{A\\ (1e+05)}&\shortstack{A\\ (1e+05)}&\shortstack{R\\ (22854)}\\ \midrule 
s298\_3\_2&\shortstack{32\\ (3e+07)}&\shortstack{A\\ (1e+05)}&\shortstack{R\\ (80883)}&\shortstack{R\\ (26491)}\\ \midrule 
s820a\_3\_2&\shortstack{37\\ (5e+07)}&\shortstack{A\\ (96212)}&\shortstack{R\\ (87997)}&\shortstack{A\\ (2e+05)}\\ \midrule 
s298\_15\_7&\shortstack{44\\ (6e+07)}&\shortstack{A\\ (1e+05)}&\shortstack{R\\ (42520)}&\shortstack{R\\ (53107)}\\ \midrule 
63.sk\_3\_64&\shortstack{58\\ (1e+08)}&DNS&\shortstack{R\\ (4393)}&\shortstack{R\\ (4677)}\\ \midrule 
s820a\_15\_7&\shortstack{79\\ (2e+08)}&\shortstack{A\\ (84310)}&\shortstack{R\\ (2e+05)}&\shortstack{R\\ (16714)}\\ \midrule 
s1488\_15\_7&\shortstack{110\\ (4e+08)}&\shortstack{A\\ (86152)}&\shortstack{R\\ (17168)}&\shortstack{R\\ (7341)}\\ \midrule 
s1488\_3\_2&\shortstack{132\\ (6e+08)}&\shortstack{A\\ (89686)}&\shortstack{A\\ (89236)}&\shortstack{R\\ (7341)}\\ \midrule 
s382\_15\_7&\shortstack{138\\ (6e+08)}&\shortstack{A\\ (92159)}&\shortstack{R\\ (2e+05)}&\shortstack{R\\ (6009)}\\ \midrule 
UserServiceImpl.sk\_8\_32&\shortstack{140\\ (6e+08)}&\shortstack{A\\ (1e+05)}&\shortstack{R\\ (1e+05)}&\shortstack{R\\ (4393)}\\ \midrule 
20.sk\_1\_51&\shortstack{144\\ (7e+08)}&DNS&\shortstack{R\\ (30895)}&\shortstack{R\\ (5146)}\\ \midrule 
s820a\_7\_4&\shortstack{167\\ (9e+08)}&\shortstack{A\\ (95566)}&\shortstack{A\\ (1e+05)}&\shortstack{R\\ (6009)}\\ \midrule 
s832a\_15\_7&\shortstack{194\\ (1e+09)}&\shortstack{A\\ (96984)}&\shortstack{R\\ (9434)}&\shortstack{R\\ (13350)}\\ \midrule 
s1488\_7\_4&\shortstack{206\\ (1e+09)}&\shortstack{A\\ (1e+05)}&\shortstack{R\\ (4677)}&\shortstack{R\\ (4627)}\\ \midrule 
s344\_15\_7&\shortstack{218\\ (2e+09)}&\shortstack{A\\ (90183)}&\shortstack{R\\ (94481)}&\shortstack{R\\ (4354)}\\ \midrule 
LoginService2.sk\_23\_36&\shortstack{232\\ (2e+09)}&\shortstack{A\\ (1e+05)}&\shortstack{R\\ (38044)}&\shortstack{R\\ (13350)}\\ \midrule 
s420\_new1\_15\_7&\shortstack{265\\ (2e+09)}&DNS&\shortstack{R\\ (19224)}&\shortstack{A\\ (3e+05)}\\ \midrule 
s349\_15\_7&\shortstack{412\\ (5e+09)}&\shortstack{A\\ (99215)}&\shortstack{R\\ (28400)}&\shortstack{R\\ (14682)}\\ \midrule 
s444\_15\_7&\shortstack{501\\ (8e+09)}&\shortstack{A\\ (1e+05)}&\shortstack{A\\ (1e+05)}&\shortstack{R\\ (26627)}\\ \midrule 
s349\_7\_4&\shortstack{603\\ (1e+10)}&\shortstack{A\\ (75555)}&\shortstack{R\\ (4284)}&\shortstack{R\\ (5150)}\\ \midrule 
s444\_7\_4&\shortstack{644\\ (1e+10)}&DNS&\shortstack{R\\ (4393)}&\shortstack{R\\ (4354)}\\ \midrule 
s420\_new1\_7\_4&\shortstack{982\\ (3e+10)}&\shortstack{A\\ (1e+05)}&\shortstack{R\\ (4354)}&\shortstack{R\\ (18473)}\\ \midrule 
s298\_7\_4&\shortstack{986\\ (3e+10)}&\shortstack{A\\ (83681)}&\shortstack{R\\ (8638)}&\shortstack{R\\ (6009)}\\ \midrule 
s420\_new1\_3\_2&\shortstack{1226\\ (5e+10)}&DNS&\shortstack{A\\ (1e+05)}&\shortstack{R\\ (5150)}\\ \midrule 
s382\_7\_4&\shortstack{1283\\ (5e+10)}&\shortstack{A\\ (92307)}&\shortstack{R\\ (26491)}&\shortstack{R\\ (7341)}\\ \midrule 
s420\_3\_2&\shortstack{1552\\ (8e+10)}&\shortstack{A\\ (1e+05)}&\shortstack{R\\ (14756)}&\shortstack{R\\ (48983)}\\ \midrule 
s1238a\_7\_4&\shortstack{1856\\ (1e+11)}&\shortstack{A\\ (95095)}&\shortstack{R\\ (5150)}&\shortstack{R\\ (7341)}\\ \midrule 
s1238a\_3\_2&\shortstack{1965\\ (1e+11)}&\shortstack{A\\ (1e+05)}&\shortstack{R\\ (28848)}&\shortstack{R\\ (4627)}\\ \midrule 
s444\_3\_2&\shortstack{2028\\ (1e+11)}&\shortstack{A\\ (1e+05)}&\shortstack{R\\ (2e+05)}&\shortstack{R\\ (9500)}\\ \midrule 
s1238a\_15\_7&\shortstack{2317\\ (2e+11)}&DNS&\shortstack{R\\ (9020)}&\shortstack{R\\ (88233)}\\ \midrule 
s420\_new\_15\_7&\shortstack{2317\\ (2e+11)}&\shortstack{A\\ (99198)}&\shortstack{R\\ (1e+05)}&\shortstack{R\\ (4393)}\\ \midrule 
30.sk\_5\_76&\shortstack{2453\\ (2e+11)}&DNS&\shortstack{R\\ (5216)}&\shortstack{R\\ (4677)}\\ \midrule 
s344\_7\_4&\shortstack{2607\\ (2e+11)}&\shortstack{A\\ (1e+05)}&\shortstack{R\\ (14170)}&\shortstack{R\\ (16818)}\\ \midrule 
s344\_3\_2&\shortstack{3300\\ (3e+11)}&\shortstack{A\\ (1e+05)}&\shortstack{R\\ (59952)}&\shortstack{R\\ (5150)}\\ \midrule 
s420\_new\_7\_4&\shortstack{3549\\ (4e+11)}&\shortstack{A\\ (82312)}&\shortstack{A\\ (96659)}&\shortstack{R\\ (49955)}\\ \midrule 
s953a\_7\_4&\shortstack{8984\\ (3e+12)}&DNS&\shortstack{A\\ (2e+05)}&\shortstack{R\\ (4627)}\\ \midrule 
s953a\_15\_7&\shortstack{10596\\ (4e+12)}&DNS&\shortstack{R\\ (11734)}&\shortstack{R\\ (59735)}\\ \midrule 
10.sk\_1\_46&\shortstack{15268\\ (7e+12)}&DNS&\shortstack{R\\ (35179)}&\shortstack{R\\ (1e+05)}\\ \midrule 
s420\_new\_3\_2&\shortstack{17449\\ (1e+13)}&\shortstack{A\\ (1e+05)}&\shortstack{R\\ (44937)}&\shortstack{R\\ (5150)}\\ \midrule 
19.sk\_3\_48&\shortstack{18253\\ (1e+13)}&DNS&\shortstack{R\\ (59014)}&\shortstack{R\\ (4627)}\\ \midrule 
s953a\_3\_2&\shortstack{20860\\ (1e+13)}&DNS&\shortstack{R\\ (51161)}&\shortstack{R\\ (1e+05)}\\ \midrule 
s641\_3\_2&\shortstack{1e+06\\(5e+16)}&DNS&\shortstack{R\\ (14454)}&\shortstack{R\\ (4627)}\\ \midrule 
ProjectService3.sk\_12\_55&\shortstack{5e+06\\(7e+17)}&DNS&\shortstack{R\\ (9020)}&\shortstack{R\\ (4393)}\\ \midrule 
71.sk\_3\_65&\shortstack{1e+07\\(3e+18)}&DNS&\shortstack{R\\ (1e+05)}&\shortstack{R\\ (4284)}\\ \midrule 
s838\_7\_4&\shortstack{1e+07\\(5e+18)}&DNS&\shortstack{R\\ (4393)}&\shortstack{R\\ (4284)}\\ \midrule 
s838\_15\_7&\shortstack{3e+07\\(3e+19)}&DNS&\shortstack{R\\ (5150)}&\shortstack{R\\ (4393)}\\ \midrule 
s713\_3\_2&\shortstack{6e+07\\(1e+20)}&DNS&\shortstack{R\\ (56386)}&\shortstack{R\\ (5827)}\\ \midrule 
s713\_7\_4&\shortstack{6e+07\\(1e+20)}&DNS&\shortstack{R\\ (5827)}&\shortstack{R\\ (37419)}\\ \midrule 
s641\_7\_4&\shortstack{9e+07\\(3e+20)}&DNS&\shortstack{R\\ (8747)}&\shortstack{A\\ (1e+06)}\\ \midrule 
s838\_3\_2&\shortstack{2e+08\\(1e+21)}&DNS&\shortstack{R\\ (9504)}&\shortstack{R\\ (4627)}\\ \midrule 
54.sk\_12\_97&\shortstack{4e+11\\(6e+27)}&DNS&\shortstack{R\\ (14012)}&\shortstack{R\\ (4627)}\\ \midrule 
\end{longtable}

\newpage
\subsection{Comparing  the runtime performance of {\barbarik} against the baseline approach}
In each of the following tables we compare the runtime of {\barbarik} against the runtime of the baseline approach. The runtime of {\barbarik} on {\reject} instances depends on which iteration the tester terminated on. The runtime of the baseline is extrapolated from the expected number of samples and the average sampling rate of the sampler. To do this we use the $\ell_1$-testing algorithm given in~\cite{batu}. In the context of this paper, the algorithm assumes black box sample access to a uniform sampler over the models of a Boolean formula $\varphi$, and the sampler under test, and requires  $O(\#\varphi^{2/3}(\eta-\varepsilon)^{-8/3}\log(\#\varphi/\delta))$ samples, where $\#\varphi$ is the model count, $(\varepsilon, \eta)$ are the closenes and farness parameters, and  $\delta$ is the confidence parameter. 
\subsubsection{Comparision with baseline for {\wSearchTreeSampler}}
\begin{center}
 \begin{longtable}{l>{\centering}m{0.2\linewidth}>{\centering\arraybackslash}m{0.2\linewidth}>{\centering\arraybackslash}m{0.2\linewidth}>{\centering\arraybackslash}m{0.16\linewidth}} 	\caption{Extended table comparing the baseline tester for {\wSearchTreeSampler} with {\barbarik}} \\ 
  \toprule 
 Benchmark & Baseline &\barbarik(s) & Speedup \\ \hline  \endhead 
 s349\_7\_4&{16457.21}&5&3428.58  \\ \hline 
s420\_new1\_7\_4&{5.4E+6}&6&8.6E+5  \\ \hline 
s298\_7\_4&{705.13}&8&94.02  \\ \hline 
s444\_7\_4&{1.1E+7}&8&1.3E+6  \\ \hline 
s832a\_7\_4&{3725.35}&10&372.53  \\ \hline 
s1488\_7\_4&{184.99}&12&15.16  \\ \hline 
s344\_7\_4&{24751.45}&15&1683.77  \\ \hline 
s420\_3\_2&{2.2E+6}&17&1.3E+5  \\ \hline 
s1238a\_7\_4&{1.4E+6}&20&66538.64  \\ \hline 
s832a\_3\_2&{2149.58}&22&98.60  \\ \hline 
s832a\_15\_7&{15121.66}&24&622.29  \\ \hline 
s838\_15\_7&{2.9E+13}&27&1.1E+12  \\ \hline 
s349\_15\_7&{16457.21}&28&587.76  \\ \hline 
s838\_7\_4&{3.7E+13}&29&1.3E+12  \\ \hline 
s382\_7\_4&{14915.27}&32&469.03  \\ \hline 
s298\_15\_7&{384.62}&32&12.09  \\ \hline 
s420\_new1\_15\_7&{4.1E+6}&33&1.3E+5  \\ \hline 
27.sk\_3\_32&{79531.43}&34&2346.06  \\ \hline 
s1238a\_15\_7&{1.8E+6}&37&49906.05  \\ \hline 
111.sk\_2\_36&{2.9E+8}&42&6.8E+6  \\ \hline 
51.sk\_4\_38&{2.0E+6}&44&45904.52  \\ \hline 
80.sk\_2\_48&{6.0E+7}&46&1.3E+6  \\ \hline 
s1488\_15\_7&{128.69}&48&2.67  \\ \hline 
s953a\_15\_7&{1.1E+9}&49&2.2E+7  \\ \hline 
s344\_3\_2&{15750.92}&51&309.45  \\ \hline 
s298\_3\_2&{229.07}&52&4.42  \\ \hline 
s838\_3\_2&{2.7E+13}&57&4.8E+11  \\ \hline 
s420\_new\_3\_2&{2.9E+6}&65&44288.35  \\ \hline 
84.sk\_4\_77&{3.4E+13}&68&5.0E+11  \\ \hline 
s641\_3\_2&{4.1E+10}&70&5.9E+8  \\ \hline 
55.sk\_3\_46&{2.0E+7}&70&2.9E+5  \\ \hline 
s349\_3\_2&{30563.39}&73&416.96  \\ \hline 
107.sk\_3\_90&{1.7E+15}&86&1.9E+13  \\ \hline 
s1238a\_3\_2&{2.2E+6}&87&25824.41  \\ \hline 
s344\_15\_7&{24751.45}&91&271.10  \\ \hline 
32.sk\_4\_38&{5.8E+5}&94&6228.23  \\ \hline 
10.sk\_1\_46&{6.5E+7}&112&5.8E+5  \\ \hline 
29.sk\_3\_45&{2.2E+8}&150&1.5E+6  \\ \hline 
s420\_new\_7\_4&{4.1E+6}&152&27272.30  \\ \hline 
s1488\_3\_2&{52.52}&163&0.32  \\ \hline 
s953a\_3\_2&{6.4E+8}&165&3.9E+6  \\ \hline 
s420\_new\_15\_7&{4.5E+6}&186&24014.34  \\ \hline 
70.sk\_3\_40&{2.9E+6}&201&14544.89  \\ \hline 
s444\_15\_7&{13470.05}&202&66.82  \\ \hline 
s420\_new1\_3\_2&{2.6E+6}&211&12084.36  \\ \hline 
s820a\_3\_2&{2189.81}&221&9.91  \\ \hline 
s444\_3\_2&{11186.45}&247&45.22  \\ \hline 
s713\_3\_2&{8.8E+10}&255&3.5E+8  \\ \hline 
109.sk\_4\_36&{6.6E+5}&269&2459.36  \\ \hline 
s820a\_7\_4&{4240.22}&277&15.33  \\ \hline 
63.sk\_3\_64&{5.8E+11}&282&2.1E+9  \\ \hline 
s641\_7\_4&{8.2E+10}&311&2.6E+8  \\ \hline 
53.sk\_4\_32&{55060.22}&313&176.08  \\ \hline 
s382\_15\_7&{33182.79}&343&96.86  \\ \hline 
s820a\_15\_7&{4154.77}&370&11.23  \\ \hline 
ProjectService3.sk\_12\_55&{1.3E+10}&458&2.9E+7  \\ \hline 
s35932\_3\_2&3.6E+2&TO &-  \\ \hline 
s35932\_7\_4&3.6E+2&TO &-  \\ \hline 
s35932\_15\_7&3.6E+2&TO &-  \\ \hline 
s953a\_7\_4&{5.7E+8}&689&8.3E+5  \\ \hline 
UserServiceImpl.sk\_8\_32&{479.33}&720&0.67  \\ \hline 
30.sk\_5\_76&{7.0E+14}&1116&6.2E+11  \\ \hline 
77.sk\_3\_44&{5.3E+6}&1687&3156.66  \\ \hline 
tableBasedAddition.sk\_240\_1024&{3.8E+14}&1832&2.1E+11  \\ \hline 
81.sk\_5\_51&{5.0E+9}&2099&2.4E+6  \\ \hline 
LoginService2.sk\_23\_36&{12951.33}&2368&5.47  \\ \hline 
20.sk\_1\_51&{1.1E+10}&2568&4.1E+6  \\ \hline 
19.sk\_3\_48&{3.1E+8}&2760&1.1E+5  \\ \hline 
17.sk\_3\_45&{4.5E+7}&3016&14948.13  \\ \hline 
71.sk\_3\_65&{4.7E+12}&4365&1.1E+9  \\ \hline 
54.sk\_12\_97&{2.7E+18}&4688&5.8E+14  \\ \hline 
\bottomrule 
 \end{longtable} 
 \end{center} 

\newpage
\subsection{{\wQuickSampler}}
\begin{center}
 \begin{longtable}{l>{\centering}m{0.2\linewidth}>{\centering\arraybackslash}m{0.2\linewidth}>{\centering\arraybackslash}m{0.2\linewidth}>{\centering\arraybackslash}m{0.16\linewidth}} 	\caption{Extended table comparing the baseline tester for {\wQuickSampler} with {\barbarik}}\\ 
  \toprule 
 Benchmark &Baseline &\barbarik(s) & Speedup \\ \hline  \endhead 
 s344\_3\_2&{24751.45}&3&8534.98 \\ \hline 
s344\_15\_7&{24751.45}&4&7071.84 \\ \hline 
s349\_7\_4&{28212.36}&4&7624.96 \\ \hline 
s298\_7\_4&{512.82}&4&119.26 \\ \hline 
s420\_new1\_3\_2&{5.1E+6}&4&1.2E+6 \\ \hline 
s420\_new\_3\_2&{2.2E+6}&4&5.1E+5 \\ \hline 
s420\_new\_15\_7&{3.5E+6}&4&7.8E+5 \\ \hline 
s382\_7\_4&{12429.39}&5&2589.46 \\ \hline 
s444\_7\_4&{51980.83}&5&10192.32 \\ \hline 
s820a\_7\_4&{2283.19}&5&430.79 \\ \hline 
s1488\_7\_4&{128.07}&6&20.99 \\ \hline 
s444\_3\_2&{8700.57}&6&1359.46 \\ \hline 
s838\_7\_4&{1.3E+13}&7&1.8E+12 \\ \hline 
27.sk\_3\_32&{48942.42}&7&6797.56 \\ \hline 
s1238a\_3\_2&{1.6E+6}&7&2.2E+5 \\ \hline 
s953a\_7\_4&{6.6E+8}&8&8.8E+7 \\ \hline 
s1488\_3\_2&{65.65}&8&8.31 \\ \hline 
s838\_3\_2&{1.9E+13}&8&2.4E+12 \\ \hline 
s1488\_15\_7&{60.56}&9&6.80 \\ \hline 
s349\_15\_7&{35265.44}&9&3833.20 \\ \hline 
s344\_7\_4&{22501.32}&9&2393.76 \\ \hline 
s349\_3\_2&{14106.18}&10&1424.87 \\ \hline 
55.sk\_3\_46&{4.5E+7}&10&4.3E+6 \\ \hline 
s1238a\_7\_4&{1.1E+6}&11&97431.59 \\ \hline 
s298\_3\_2&{534.49}&11&46.89 \\ \hline 
s832a\_7\_4&{4139.28}&12&344.94 \\ \hline 
111.sk\_2\_36&{5.2E+5}&12&41613.34 \\ \hline 
s838\_15\_7&{2.6E+13}&12&2.1E+12 \\ \hline 
s420\_new1\_7\_4&{2.2E+6}&13&1.7E+5 \\ \hline 
s832a\_15\_7&{13861.52}&14&1011.79 \\ \hline 
UserServiceImpl.sk\_8\_32&{326.81}&14&23.68 \\ \hline 
s382\_15\_7&{27149.56}&15&1859.56 \\ \hline 
53.sk\_4\_32&{91767.04}&16&5595.55 \\ \hline 
s820a\_15\_7&{5665.59}&17&335.24 \\ \hline 
84.sk\_4\_77&{2.1E+13}&18&1.2E+12 \\ \hline 
51.sk\_4\_38&{1.8E+6}&19&91363.08 \\ \hline 
s444\_15\_7&{14817.06}&19&763.77 \\ \hline 
109.sk\_4\_36&{6.6E+5}&20&33425.00 \\ \hline 
107.sk\_3\_90&{1.6E+15}&21&7.4E+13 \\ \hline 
71.sk\_3\_65&{1.3E+12}&27&5.0E+10 \\ \hline 
s641\_3\_2&{2.8E+10}&28&1.0E+9 \\ \hline 
s298\_15\_7&{1153.85}&30&38.98 \\ \hline 
32.sk\_4\_38&{1.2E+6}&34&36689.91 \\ \hline 
s420\_3\_2&{4.5E+6}&34&1.3E+5 \\ \hline 
s420\_new\_7\_4&{3.5E+6}&36&96896.05 \\ \hline 
80.sk\_2\_48&{2.1E+8}&37&5.7E+6 \\ \hline 
s832a\_3\_2&{2149.58}&45&47.66 \\ \hline 
19.sk\_3\_48&{4.5E+8}&50&9.0E+6 \\ \hline 
63.sk\_3\_64&{2.1E+11}&51&4.0E+9 \\ \hline 
17.sk\_3\_45&{8.3E+7}&55&1.5E+6 \\ \hline 
s713\_3\_2&{9.4E+10}&56&1.7E+9 \\ \hline 
s953a\_15\_7&{6.7E+8}&79&8.5E+6 \\ \hline 
20.sk\_1\_51&{4.0E+9}&82&4.8E+7 \\ \hline 
70.sk\_3\_40&{4.3E+6}&101&42475.10 \\ \hline 
s1238a\_15\_7&{1.0E+6}&107&9614.31 \\ \hline 
10.sk\_1\_46&{7.1E+7}&128&5.5E+5 \\ \hline 
s953a\_3\_2&{3.4E+8}&132&2.6E+6 \\ \hline 
s820a\_3\_2&{1167.90}&137&8.54 \\ \hline 
30.sk\_5\_76&{3.0E+14}&210&1.4E+12 \\ \hline 
ProjectService3.sk\_12\_55&{6.4E+9}&219&2.9E+7 \\ \hline 
LoginService2.sk\_23\_36&{12692.30}&229&55.52 \\ \hline 
s420\_new1\_15\_7&{3.2E+6}&232&13726.91 \\ \hline 
77.sk\_3\_44&{1.2E+7}&409&30125.88 \\ \hline 
29.sk\_3\_45&{1.3E+8}&658&2.0E+5 \\ \hline 
54.sk\_12\_97&{4.0E+17}&690&5.8E+14 \\ \hline 
s641\_7\_4&{6.8E+10}&1117&6.1E+7 \\ \hline 
s35932\_15\_7&{1.4E+356}&1182&1.2E+353 \\ \hline 
tableBasedAddition.sk\_240\_1024&{3.0E+13}&1430&2.1E+10 \\ \hline 
s35932\_7\_4&{1.2E+356}&2227&5.5E+352 \\ \hline 
s35932\_3\_2&{1.1E+356}&2346&4.5E+352 \\ \hline 
81.sk\_5\_51&{2.0E+9}&2461&8.3E+5 \\ \hline 
\bottomrule 
 \end{longtable} 
 \end{center}

\newpage
\subsection{{\wWeightGen}}
\begin{center}
 \begin{longtable}{l>{\centering}m{0.2\linewidth}>{\centering\arraybackslash}m{0.2\linewidth}>{\centering\arraybackslash}m{0.2\linewidth}>{\centering\arraybackslash}m{0.16\linewidth}} 	\caption{Extended table comparing the baseline tester for {\wWeightGen} with {\barbarik}}\\ 
  \toprule 
 Benchmark & Baseline &\barbarik(s) & Speedup \\ \hline  \endhead 
 s1488\_3\_2&{229.78}&6648&0.03 \\ \hline 
s298\_7\_4&{7564.11}&10758&0.70 \\ \hline 
s1488\_15\_7&{643.45}&11493&0.06 \\ \hline 
s298\_15\_7&{2948.72}&12325&0.24 \\ \hline 
s349\_7\_4&{1.8E+06}&12858&136.40 \\ \hline 
s820a\_15\_7&{48724.11}&14070&3.46 \\ \hline 
s344\_15\_7&{3.8E+05}&14074&27.18 \\ \hline 
s1488\_7\_4&{853.78}&15049&0.06 \\ \hline 
s820a\_7\_4&{42728.33}&16124&2.65 \\ \hline 
s349\_15\_7&{3.9E+05}&17690&21.80 \\ \hline 
s382\_7\_4&{9.7E+05}&21785&44.45 \\ \hline 
s349\_3\_2&{3.0E+05}&22395&13.54 \\ \hline 
s832a\_15\_7&{5.6E+05}&23036&24.45 \\ \hline 
s420\_new\_7\_4&{4.0E+09}&24092&1.7E+5 \\ \hline 
s344\_7\_4&{1.7E+06}&26423&64.55 \\ \hline 
51.sk\_4\_38&{2.7E+09}&26612&1.0E+5 \\ \hline 
s820a\_3\_2&{2.3E+05}&27408&8.47 \\ \hline 
s298\_3\_2&{2061.62}&30262&0.07 \\ \hline 
s344\_3\_2&{5.0E+05}&32378&15.29 \\ \hline 
s1238a\_7\_4&{1.5E+09}&33689&45408.69 \\ \hline 
s832a\_7\_4&{76990.55}&34315&2.24 \\ \hline 
s382\_15\_7&{1.0E+07}&39024&263.98 \\ \hline 
s1238a\_3\_2&{7.1E+08}&40406&17575.38 \\ \hline 
s420\_new\_15\_7&{4.9E+09}&40725&1.2E+5 \\ \hline 
27.sk\_3\_32&{7.4E+06}&41997&176.26 \\ \hline 
s832a\_3\_2&{74844.43}&42696&1.75 \\ \hline 
UserServiceImpl.sk\_8\_32&{21547.88}&45090&0.48 \\ \hline 
32.sk\_4\_38&{4.9E+08}&45126&10872.88 \\ \hline 
s420\_new1\_7\_4&{2.8E+08}&48911&5639.38 \\ \hline 
s444\_3\_2&{1.9E+06}&55017&34.61 \\ \hline 
LoginService2.sk\_23\_36&{1.3E+06}&56229&22.38 \\ \hline 
s420\_3\_2&{2.3E+09}&68048&33247.50 \\ \hline 
53.sk\_4\_32&{2.2E+07}&70590&312.87 \\ \hline 
s420\_new\_3\_2&{1.2E+10}&75284&1.6E+5 \\ \hline 
\bottomrule 
 \end{longtable} 
 \end{center} 

\newpage
\subsection{Number of samples required for baseline approach}
\begin{center}
 \begin{longtable}{l>{\centering\arraybackslash}m{0.16\linewidth}}   \caption{Number of samples required for baseline tester} \\ 
  \toprule 
 Benchmark  & Number of samples \\ \hline  \endhead 
 s344\_3\_2&2E+6 \\ \hline 
s344\_15\_7&2E+6 \\ \hline 
s349\_7\_4&2E+6 \\ \hline 
s298\_7\_4&6E+4 \\ \hline 
s420\_new1\_3\_2&3E+8 \\ \hline 
s420\_new\_3\_2&3E+8 \\ \hline 
s420\_new\_15\_7&3E+8 \\ \hline 
s382\_7\_4&1E+6 \\ \hline 
s444\_7\_4&4E+6 \\ \hline 
s820a\_7\_4&3E+5 \\ \hline 
s1488\_7\_4&1E+4 \\ \hline 
s444\_3\_2&1E+6 \\ \hline 
s838\_7\_4&2E+15 \\ \hline 
27.sk\_3\_32&6E+6 \\ \hline 
s1238a\_3\_2&1E+8 \\ \hline 
s953a\_7\_4&4E+10 \\ \hline 
s1488\_3\_2&7E+3 \\ \hline 
s838\_3\_2&2E+15 \\ \hline 
s1488\_15\_7&8E+3 \\ \hline 
s349\_15\_7&2E+6 \\ \hline 
s344\_7\_4&2E+6 \\ \hline 
s349\_3\_2&2E+6 \\ \hline 
55.sk\_3\_46&2E+9 \\ \hline 
s1238a\_7\_4&1E+8 \\ \hline 
s298\_3\_2&4E+4 \\ \hline 
s832a\_7\_4&4E+5 \\ \hline 
111.sk\_2\_36&3E+7 \\ \hline 
s838\_15\_7&2E+15 \\ \hline 
s420\_new1\_7\_4&3E+8 \\ \hline 
s832a\_15\_7&1E+6 \\ \hline 
UserServiceImpl.sk\_8\_32&2E+4 \\ \hline 
s382\_15\_7&3E+6 \\ \hline 
53.sk\_4\_32&6E+6 \\ \hline 
s820a\_15\_7&4E+5 \\ \hline 
84.sk\_4\_77&1E+15 \\ \hline 
51.sk\_4\_38&8E+7 \\ \hline 
s444\_15\_7&1E+6 \\ \hline 
109.sk\_4\_36&4E+7 \\ \hline 
107.sk\_3\_90&7E+16 \\ \hline 
71.sk\_3\_65&5E+13 \\ \hline 
s641\_3\_2&3E+12 \\ \hline 
s298\_15\_7&6E+4 \\ \hline 
32.sk\_4\_38&7E+7 \\ \hline 
s420\_3\_2&3E+8 \\ \hline 
s420\_new\_7\_4&3E+8 \\ \hline 
80.sk\_2\_48&6E+9 \\ \hline 
s832a\_3\_2&2E+5 \\ \hline 
19.sk\_3\_48&1E+10 \\ \hline 
63.sk\_3\_64&5E+12 \\ \hline 
17.sk\_3\_45&2E+9 \\ \hline 
s713\_3\_2&6E+12 \\ \hline 
s953a\_15\_7&5E+10 \\ \hline 
20.sk\_1\_51&7E+10 \\ \hline 
70.sk\_3\_40&2E+8 \\ \hline 
s1238a\_15\_7&1E+8 \\ \hline 
10.sk\_1\_46&6E+9 \\ \hline 
s953a\_3\_2&4E+10 \\ \hline 
s820a\_3\_2&1E+5 \\ \hline 
30.sk\_5\_76&2E+15 \\ \hline 
ProjectService3.sk\_12\_55&2E+11 \\ \hline 
LoginService2.sk\_23\_36&1E+5 \\ \hline 
s420\_new1\_15\_7&3E+8 \\ \hline 
77.sk\_3\_44&3E+8 \\ \hline 
29.sk\_3\_45&3E+9 \\ \hline 
54.sk\_12\_97&4E+18 \\ \hline 
s641\_7\_4&5E+12 \\ \hline 
s35932\_15\_7&1E+357 \\ \hline 
tableBasedAddition.sk\_240\_1024&1E+15 \\ \hline 
s35932\_7\_4&1E+357 \\ \hline 
s35932\_3\_2&1E+357 \\ \hline 
81.sk\_5\_51&4E+10 \\ \hline 
\bottomrule 
 \end{longtable} 
 \end{center}

\end{document}